\documentclass[prl,aps,twocolumn,superscriptaddress]{revtex4}
\usepackage{amsfonts}
\usepackage{amsmath}
\usepackage{amssymb}
\usepackage{amsthm}
\usepackage{bm}
\usepackage{chngcntr}
\usepackage{enumerate}
\usepackage{enumitem}
\usepackage{esint}
\usepackage{eucal}
\usepackage{float}
\usepackage{framed}
\usepackage{graphicx}
\usepackage[pdftex,colorlinks=true]{hyperref}
\usepackage{latexsym}
\usepackage{listings}
\usepackage{mathtools}
\usepackage{psfrag}
\usepackage{subfigure}
\usepackage{xcolor}
\usepackage[all]{xy}
\usepackage[english]{babel}

\newcommand{\ii}{\mathrm{i}}

\newcommand{\ket}[1]{|#1\rangle}
\newcommand{\bra}[1]{\langle#1|}

\newcommand{\overlap}[2]{\langle #1|#2 \rangle}

\newcommand{\Tr}{\mathrm{Tr}}

\newtheorem{theorem}{Theorem}[section]
\newtheorem{corollary}{Corollary}[theorem]
\newtheorem{lemma}[theorem]{Lemma}

\newtheorem{definition}[theorem]{Definition}
\newtheorem{property}[theorem]{Property}

\begin{document}
\title{Multi-Layer Restricted Boltzmann Machine Representation of 1D Quantum Many-Body Wave Functions}
\date{\today}

\author{Huan He}
\affiliation{Department of Physics, Princeton University, Princeton, New Jersey 08544, USA}

\author{Yunqin Zheng}
\affiliation{Department of Physics, Princeton University, Princeton, New Jersey 08544, USA}

\author{B. Andrei Bernevig}
\affiliation{Department of Physics, Princeton University, Princeton, New Jersey 08544, USA}
\affiliation{Physics Department, Freie Universitat Berlin, Arnimallee 14, 14195 Berlin, Germany}
\affiliation{Max Planck Institute of Microstructure Physics, 06120 Halle, Germany}

\author{German Sierra}
\affiliation{
Instituto de F\'{\i}sica Te\'orica, UAM-CSIC, Universidad Aut\'onoma de Madrid, Cantoblanco, Madrid, Spain}

\begin{abstract}
	We consider representing two classes of 1D quantum wave functions of spin systems, including the AKLT and CFT correlator wave functions, in terms of multi-layer restricted Boltzmann machines. In our prescription, the AKLT wave function can be exactly represented by a 2-layer restricted Boltzmann machine with five hidden spins per visible spin. The construction can be generalized to prove that any MPS wave function on $N$ unit cells with finite bond dimension can be approximated by a 2-layer restricted Boltzmann machine with $\mathcal{O}(N)$ hidden spins within an error which scales linearly with $N$. The Haldane-Shastry wave function or a chiral boson CFT correlator wave function, as any Jastrow type of wave functions, can be exactly written as a 1-layer Boltzmann machine with $\mathcal{O}(N^2)$ hidden spins and $N$ visible spins. Applying the cumulant expansion, we further find that the chiral boson CFT correlator wave function (with small  vertex operator conformal dimension $\alpha$, i.e., $\alpha<0.1$) can be approximated, within 99.9\% accuracy up to 22 visible spins, by a 1-layer RBM with $\mathcal{O}(N)$ hidden spins. The cumulant expansion also leads us to a physically inspiring result in which the hidden spins of the restricted Boltzmann machine can be interpreted as the conformal tower of the chiral boson CFT on the cylinder.
\end{abstract}

\maketitle

\paragraph*{Introduction}

Restricted Boltzmann machine (RBM) states have been recently attracting the attention of the condensed matter community as a new numerical tool for quantum systems\cite{Carleo2017Solving,torlai2018neural,carleo2018constructing,jia2018efficient,Kaubruegger2018Chiral,Glasser2018Neural,Clark2018Unifying,Nomura2017Restricted,Rao2018Identifying,morningstar2017deep,gao2017efficient,Deng2017Quantum,Chen2018Equivalence,Aoki2016Restricted,Tubiana2017Emergence,Huang2017Accelerated,2018PhRvB..97c5116C, 2018arXiv181002352L, 2017arXiv170106246H,PhysRevB.91.155150,PhysRevLett.61.2376}. The inspiration of RBM states comes from the fields of deep learning and artificial intelligence: neural networks. One of the earliest neural networks is the RBM. In this paper, we focus on a simple generalization of the RBM, the \emph{$k$-layer RBM}. A $k$-layer RBM refers to the following type of functions:
\begin{equation}\label{eq.RBM}
\begin{split}
&f_{\mathrm{k-RBM}}(s_1,s_2,\ldots,s_N) = \sum_{h^1_{1},\ldots,h^1_{m_1}} \sum_{h^2_{1},\ldots,h^2_{m_2}} \ldots \sum_{h^k_{1},\ldots,h^k_{m_k}} \\
& \exp\left( \sum_{p_0, p_1} W^1_{p_0p_1} s_{p_0}h^1_{p_1} + \sum_{p_0} s_{p_0} a_{p_0} \right. \\
& \left. + \sum_{i=2}^{k} \sum_{p_{i-1}, p_i} W^i_{p_{i-1}p_i} h^{i-1}_{p_{i-1}} h^i_{p_i} + \sum_{i=1}^{k}\sum_{p_i} h^i_{p_i} b^i_{p_i}\right)
\end{split}
\end{equation}
where (i) $i\in \lbrace 1, ..., k\rbrace $ labels the hidden layer; (ii) $s_{p_0}$ $(1\leq p_0\leq n)$ are the \textit{visible spins}; (iii) $h^i_{p_i}$ $(1\leq p_i\leq m_i)$ are the \emph{hidden spins} in the $i$-th layer; (iv) $W^i_{p_ip_{i-1}}$ are complex valued weights; (v) $a_{p_0}$ and $b^i_{p_i}$ are complex valued biases for the visible and $i$-th layer hidden spins respectively. An important feature of the $k$-layer RBM is that the hidden spins in the $i$-th layer only couple to the hidden spins in the $(i+1)$-th or $(i-1)$-th layers, and they do not couple to the hidden spins within the same layer, although integrating out the $i$-th layer will give spin-spin coupling in the $(i-1)$-th and $(i+1)$-th layer. See Fig.~\ref{fig.RBM} for an graphical representation of a 3-layer RBM. In particular, the 1-layer RBM is:
\begin{equation}\label{eq.RBM.1layer}
\begin{split}
&f_{1-\mathrm{RBM}}(s_1,s_2,\ldots,s_N) = \\
&\sum_{h_{1},\ldots,h_{M}} 
\exp\left( \sum_{i,j} W_{ij} s_{i}h_{j} + \sum_{i} s_{i} a_{i} + \sum_{j} h_{j} b_{j} \right) \\
\end{split}
\end{equation}
The values of the visible and hidden spins can vary depending on the systems. 

\begin{figure}
	\centering
	\includegraphics[width=0.8\columnwidth]{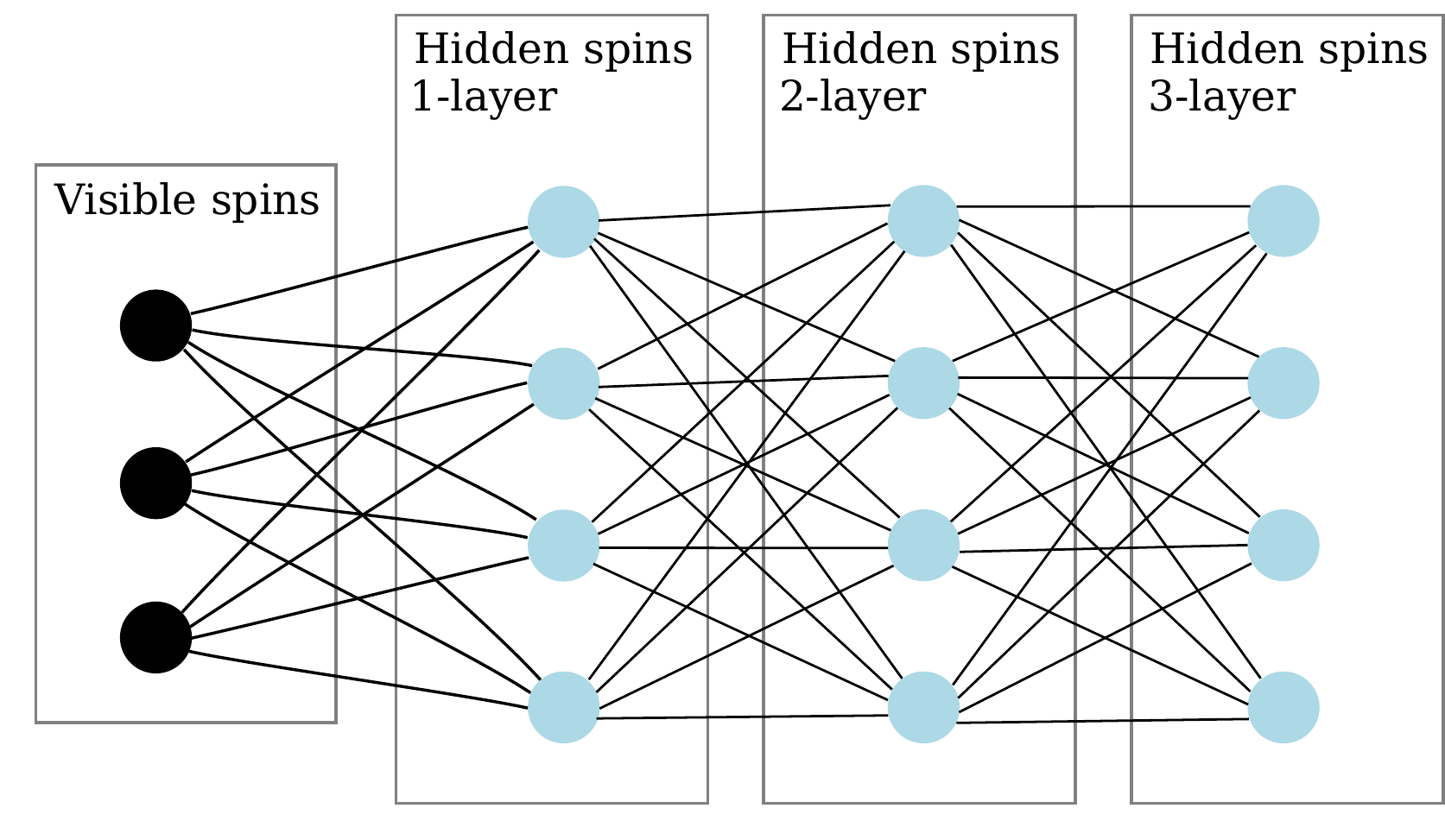}
	\caption{A neural network diagram representation of a 3-layer fully connected RBM. The black dots are the visible spins while the blue ones are hidden. Each of the black lines connecting two dots represents the $W$ weight in Eq.~\eqref{eq.RBM}. The biases are not explicit in this diagram. The first hidden layer is connected to the visible layer and the second hidden layer, and the second hidden layer is connected to the first and the third one. The last hidden layer only connect with the second hidden layer. The fully connected RBM refers to an RBM where each of the spins in a layer is connected to \emph{all} the spin in the nearest neighbor layers. }
	\label{fig.RBM}
\end{figure}

Previous studies \cite{Chen2018Equivalence, 2018arXiv180908631Z, zheng2018restricted} mostly focus on the numerical applications of 1-layer RBM to quantum many-body systems. Only few analytical studies demonstrate the power of $k$-layer RBM's. In Ref.~\cite{Chen2018Equivalence, 2018arXiv180908631Z}, efficient numerical algorithms were proposed to construct 1-layer RBM states from tensor network states or from stabilizer codes. In Ref.~\cite{zheng2018restricted}, finitely connected\footnote{Here, finitely connected means that the weights between visible spins and hidden spins are non-vanishing only when the distance between the visible and hidden spins are within a finite range. } 1-layer RBM states with minimal number of hidden spins have been analytically constructed for a family of 1D stabilizer codes. The purpose of this paper is to represent more generic wave functions beyond stabilizer codes in terms of $k$-layer RBM's. In particular, we focus on two examples: the AKLT model and the free-boson CFT correlator states. As proved in Ref.~\cite{zheng2018restricted}, a necessary condition of the existence of a finitely connected 1-layer RBM to \emph{exactly} represent a 1D wave function is that the matrices of its matrix product state (MPS)\cite{Hastings2007Area, 1987PhRvL..59..799A, fannes1992, PhysRevLett.69.2863, PhysRevLett.75.3537, 2006quant.ph..8197P, verstraete2008matrix, 2008PhRvL.100c0504S, schollwock2011}, which can be immediately obtained from the RBM, are of rank 1. However, this condition does not hold true for AKLT model, since one of its MPS matrix is of rank 2.\cite{1987PhRvL..59..799A} We will show in this paper that the spin-1 AKLT model can be exactly represented by a 2-layer RBM. Furthermore, we prove in section D of the supplementary material that any MPS can be approximated, up to an arbitrary small error that scales linearly with the system size $N$, by a 2-layer RBM.

\paragraph*{2-Layer RBM for the AKLT State}

For clarity, we denote $\textstyle \lbrace \ket{\frac{1}{2}}, \ket{-\frac{1}{2}} \rbrace$ as the basis of a spin-$\frac{1}{2}$ state, and $\textstyle \lbrace \ket{1}, \ket{0}, \ket{-1} \rbrace$ as the basis of a spin-$1$ state. Consider a spin-$1$ chain with periodic boundary condition containing $N$ unit cells. The parent Hamiltonian of the AKLT model is 
$
H=\sum_{i=1}^{N} \left[ \mathbf{S}_{i} \cdot \mathbf{S}_{i+1} +\frac{1}{3} \left( \mathbf{S}_{i} \cdot \mathbf{S}_{i+1}\right)^2\right]
$
where $S_i^a$ $(a=1,2,3)$ are three spin-1 operators for the $i$-th spin. The AKLT ground state can be represented as an MPS 
\begin{eqnarray}
|\mathrm{AKLT}\rangle=\sum_{\lbrace s_i\rbrace =\lbrace -1, 0, 1\rbrace ^{N}} \Tr\bigg(\prod_{i=1}^{N}T^{s_i}\bigg) |\lbrace s_i\rbrace \rangle
\end{eqnarray}
where 
\begin{eqnarray}\label{Eq.AKLTMPS}
T^{-1}= \sqrt{\frac{2}{3}}\sigma^+, T^0= -\sqrt{\frac{1}{3}}\sigma^3, T^1=-\sqrt{\frac{2}{3}}\sigma^-
\end{eqnarray}
$\textstyle \sigma^{\pm}= (\sigma^1\pm i \sigma^2)/2$, and $\textstyle \sigma^{1,2,3}$ are the 2-by-2 Pauli matrices. The basis $|{s_i}\rangle$ is a direct product state $|{s_i}\rangle= \otimes_{i=1}^{N}|s_i\rangle$. The AKLT ground state and its MPS representation can be constructed by starting with a singlet chain of spin-$\frac{1}{2}$'s (to be specified below) and project the spin-$\frac{1}{2}$'s into spin-1's \cite{1987PhRvL..59..799A,schollwock2011}. As we will find, a 2-layer RBM representation of the ground state can be obtained using the same construction. 

\paragraph{1. A Singlet Chain as an RBM:}

Consider a translational invariant chain with $N$ unit cells where each unit cell 
contains two spin-$\frac{1}{2}$'s, labeled by $a_i, b_i=\pm \frac{1}{2}$, $i=1, ..., N$. A singlet chain state is a tensor product of pairs of singlets $|\psi\rangle_{i,i+1}$ across adjacent unit cells, i.e., $\textstyle \ket{\psi}=\bigotimes_{i=1}^{N} \ket{\psi}_{i,i+1}$
with
\begin{eqnarray}\label{eq.singlet}
\begin{split}
&\ket{\psi}_{i,i+1}	\\
=& \frac{\ii}{\sqrt{2}} \left( \ket{b_i=\frac{1}{2},a_{i+1}=-\frac{1}{2}}-\ket{b_i=-\frac{1}{2},a_{i+1}=\frac{1}{2}} \right) \\
=& \frac{1}{2\sqrt{2}} 
\sum_{b_i,a_{i+1}=\pm \frac{1}{2}} \sum_{h_{i,i+1}=0}^{1} e^{\ii \pi h_{i,i+1} (b_i+a_{i+1}) + \ii \pi b_i} \ket{b_i,a_{i+1}}
\end{split}
\end{eqnarray}
where we introduce a hidden spin $h_{i,i+1}=0,1$ for each pair of adjacent unit cells to fit the coefficient in the form of an RBM. Suppressing the overall normalization factor, the singlet chain $\ket{\psi}$ can be compactly written by taking the product of Eq.~\eqref{eq.singlet} as:
\begin{eqnarray}\label{eq.singlet.chain.BM}
\begin{split}
\ket{\psi}
=& \sum_{\lbrace a_i,b_i,h_{i,i+1}\rbrace } e^{\ii \sum_{i} \left( \pi h_{i,i+1} (b_i+a_{i+1}) + \pi b_i \right)}
\ket{\lbrace a_i,b_i\rbrace }	\\
\end{split}
\end{eqnarray}
The summation is taken over all possible configurations of $a_i$, $b_i$ and $h_{i,i+1}$.

\paragraph{2. Projection as a 1-layer RBM:}

The second step is to project the two spin-$\frac{1}{2}$'s on each site to a spin-$1$ via the projector $\mathbb{P}$ 
\begin{equation}\label{Eq.Proj}
\mathbb{P}=\bigotimes_{i=1}^{N}\mathbb{P}_i =\bigotimes_{i=1}^{N} \bigg(\sum_{s_i=-1,0,1} \sum_{a_i,b_i=\pm \frac{1}{2}} P^{s_i}_{a_i,b_i} \ket{s_i}\bra{a_i,b_i}\bigg)
\end{equation}
where the Clebsch-Gordan (CG) coefficients $P^{s}_{a,b}$ are
\begin{eqnarray}\label{eq.project_tensor}
\begin{split}
P^{1}_{\frac{1}{2},\frac{1}{2}}&=P^{-1}_{-\frac{1}{2},-\frac{1}{2}}=1,	\quad
P^{0}_{\frac{1}{2},-\frac{1}{2}}=P^{0}_{-\frac{1}{2},\frac{1}{2}}=\frac{1}{\sqrt{2}},	\\
P^{s}_{a,b}&=0, \quad\text{otherwise}.
\end{split}
\end{eqnarray}
Relegating the details to section A of the supplementary material, we find that the CG coefficients can be expressed as a 1-layer RBM,
\begin{eqnarray}\label{eq.projector.as.BM}
\begin{split}
P^{s}_{a,b}
= \frac{1}{6} \sum_{h=-1}^{1}\sum_{h^\prime=\pm \frac{1}{2}} e^{\ii\left( \frac{2\pi}{3} h (a+b-s)+ \frac{\pi}{2}h^\prime(a-b) \right)} \\
\end{split}
\end{eqnarray}
Substituting Eq.~\eqref{eq.projector.as.BM} into Eq.~\eqref{Eq.Proj}, expanding the direct product and suppressing the overall normalization factor, we find the projector $\mathbb{P}$ as
\begin{eqnarray}\label{eq.projector.BM}
\begin{split}
\mathbb{P} \propto \sum_{\lbrace s_i, a_i,b_i,h_i,h^\prime_i\rbrace}
& e^{\pi\ii \sum_i \left( \frac{2}{3} h_i (a_i+b_i-s_i)+ \frac{1}{2}h_i^\prime(a_i-b_i) \right)}	\\
& \ket{\lbrace s_i \rbrace}\bra{\lbrace a_i, b_i \rbrace}	\\
\end{split}
\end{eqnarray}
where $|\{s_i\}\rangle$ is defined below Eq.~\eqref{Eq.AKLTMPS} and $\ket{\lbrace a_i, b_i \rbrace}= \otimes_{i=1}^N |a_i, b_i\rangle$. 
We refer to the projector Eq.~\eqref{eq.projector.BM} as the RBM projector. 

\paragraph{3. AKLT State:}

\begin{figure}[t]
	\centering
	\includegraphics[width=0.8\columnwidth]{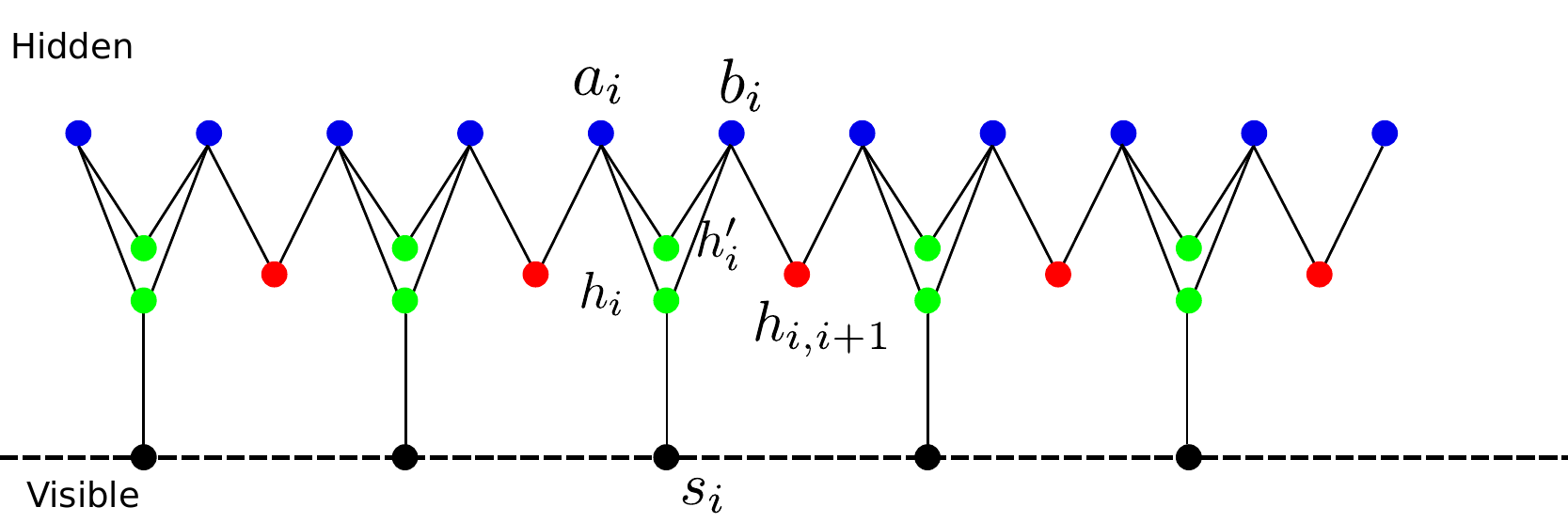}
	\caption{The illustration for the 2-layer RBM state of the spin-$1$ AKLT wave function. The black dots are the visible spin-1's. The colored dots are the hidden spins. The weights are the solid black lines. The dashed black line is the 1D lattice.}
	\label{fig.AKLT.spin1.BM}
\end{figure}

We finally combine the results in Eqs.~\eqref{eq.singlet.chain.BM} and \eqref{eq.projector.BM} to express the spin-$1$ AKLT state in terms of a 2-layer RBM state:
\begin{equation}\label{eq.AKLT.spin1.BM}
\begin{split}
&\mathbb{P} \ket{\psi}
\propto 
\sum_{\lbrace s_i,a_i,b_i,h_{i,i+1},h_i,h^\prime_i\rbrace } \\
&e^{\pi\ii \sum_{i} \left( h_{i,i+1} (b_i+a_{i+1}) + b_i + \frac{2}{3} h_i (a_i+b_i-s_i)+ \frac{1}{2}h_i^\prime(a_i-b_i) \right)} \ket{\lbrace s_i\rbrace}
\end{split}
\end{equation}
where the last expression is a 2-layer RBM state. $h_{i}, h_{i}^{\prime}$ and $h_{i,i+1}$ are the hidden spins in the first layer, while $a_{i},b_{i}$ are the hidden spins in the second layer. See Fig.~\ref{fig.AKLT.spin1.BM} for a graphical representation of the 2-layer RBM.
Applying the procedure in Ref.~\cite{zheng2018restricted} (see also section B of the supplementary material), the RBM Eq.~\eqref{eq.AKLT.spin1.BM} can be cast into the form of an MPS, 
\begin{equation}\label{Eq.MPS}
\begin{split}
&\tilde{T}^{s_i}_{h_{i-1, i} h_{i, i+1}}
\propto \sum_{a_i, b_i, h_i, h_i^\prime } 
\exp\bigg[ \pi\ii \bigg( h_{i-1,i} a_{i} + \\
& h_{i,i+1} b_i + b_i + \frac{2}{3} h_i (a_i+b_i-s_i)+ \frac{1}{2}h_i^\prime(a_i-b_i)
\bigg) \bigg]
\end{split}
\end{equation}
After a similarity transformation on the virtual indices,
\begin{eqnarray}
\tilde{T}^{s_i}\to U\cdot \tilde{T}^{s_i}\cdot U^{-1}, ~~
U= \frac{1}{\sqrt{2}}
\begin{pmatrix}
1 & \ii\\
1 & -\ii
\end{pmatrix}
\end{eqnarray}
we obtain the standard MPS matrices of the AKLT state Eq.~\eqref{Eq.AKLTMPS}.

\paragraph*{MPS by 2-Layer RBM}

The above procedure of expressing the AKLT state as a 2-layer RBM can be generalized for any MPS with finite bond dimension. In section D of the supplementary material, we prove that given an arbitrary small error $\epsilon$ and arbitrary MPS $\mathrm{MPS}(\{s_i\})= \prod_{\{s_i\}}\Tr(T^{s_i})$ with finite bond dimension on a system with $N$ unit cells,  there exists a 2-layer RBM such that
\begin{equation}\label{Eq.appro}
\max_{\{s_i\}}|\mathrm{RBM}(\{s_i\})- \mathrm{MPS}(\{s_i\})|\simeq \mathcal{O}(N) \epsilon
\end{equation}
where $\mathcal{O}(N)$ is a coefficient that scales linearly with $N$.  
We prove Eq.~\eqref{Eq.appro} based on the representing power theorem of 1-layer RBM for (real) probability distributions of binary numbers\cite{le2008representational}. We leave the proof to the supplementary material. 

\paragraph*{Correlator States by RBM}

In this section, we construct the \emph{exact} 1-layer RBM  representions for a series of chiral boson CFT correlator wave functions. 
For simplicity, we adopt a slightly different convention for labeling the spin basis: 
$|s\rangle= |\pm\rangle$. 
The vertex operator of the chiral boson CFT \cite{Cirac2010Infinite}\footnote{The standard vertex operator of chiral boson CFT is $\exp(i \sqrt{\alpha} \phi(z))$ which does not contain the prefactor $s_i$. The vertex operator in Eq.~\eqref{eq.chiralboson.vertexoperators} is more like Knizhnik-Zamolodchikov spin vertices. However, for simplicity, we still denote Eq.~\eqref{eq.chiralboson.vertexoperators} as the vertex operator of chiral boson CFT. } is
\begin{eqnarray}\label{eq.chiralboson.vertexoperators}
A^i_{s_i} = \chi_{s_i} : e^{\ii s_i \sqrt{\alpha} \phi(z_i)}:, 
\quad i=1,2,\ldots,N.
\end{eqnarray}
We explain the notations in Eq.~\eqref{eq.chiralboson.vertexoperators}. (i) $N$ is the system size and we assume that $N$ is even; (ii) $\phi(z)$ is the holomorphic part of the free boson field, where $z=x+\ii y$ is the complex coordinate; (iii) $s_i=\pm 1$ labels the spin basis; (iv) $\alpha$ is a positive real number encoding the conformal dimension $h$ of $A^i_{s_i}$, i.e., $h=\frac{1}{2} \alpha$; (v) the Marshall sign factor $\chi_{s_j}$ is defined as $\chi_{s_j} = 1$ for  even $j$, and $\chi_{s_j}=s_j$ for odd $j$; (vi) $:\mathcal{O}:$ denotes the normal ordering of the operator $\mathcal{O}$ on the decomposition of $\phi$ into normal modes. The correlator of the vertex operators Eq.~\eqref{eq.chiralboson.vertexoperators} is the wave function:
\begin{eqnarray}\label{eq.CFT.correlator}
\begin{split}
&\psi_{\mathrm{CFT}}(s_1,s_2,\ldots,s_{N})
\propto \langle A^1_{s_1}A^2_{s_2} \ldots A^{N}_{s_{N}} \rangle_{\mathrm{CFT}}	\\
\propto& \delta\left( \sum_{i=1}^{N} s_i\right)  \prod_{i ~\mathrm{odd}} s_i \prod_{i>j} (z_i-z_j)^{\alpha s_i s_j}
\end{split}
\end{eqnarray}
where $ \textstyle \delta\left( \sum_{i=1}^{N} s_i\right)=1 $ if $\textstyle \sum_{i=1}^{N} s_i=0$; otherwise $\textstyle \delta\left( \sum_{i=1}^{N} s_i\right) = 0$, i.e, constraining the total $s_z$ to be 0. If $z_i$'s are restricted to be the coordinates of a spin chain of $N$ sites with periodic boundary condition, i.e.,
$\textstyle z_n=e^{\frac{2\pi\ii}{N} n},\;n\in\lbrace 0, 1, ..., N-1\rbrace $, then Eq.~\eqref{eq.CFT.correlator} reduces to
\begin{eqnarray}\label{eq.CFT.wavefunction}
\begin{split}
&\psi_{\mathrm{CFT}}(s_1,s_2,\ldots,s_{N}) \propto 
\delta\left( \sum_{i=1}^{N} s_i\right) \prod_{i ~\mathrm{odd}} s_i \\
& \times 
\prod_{m>n} \left[\sin\left(\pi\frac{m-n}{N}\right)\right]^{\alpha s_m s_n}
\end{split}
\end{eqnarray}
In particular, when $\textstyle \alpha=\frac{1}{2}$, $\psi_{\mathrm{CFT}}$ reduces to the Haldane-Shastry wave function \cite{Cirac2010Infinite}\footnote{To see its connection to the HS wave function, it is more convenient to change the spin basis $s_i$ to the occupation number basis $n_i$, via $s_i=1-2a_i$. $a_i=0,1$. Since $\sum_i s_i=0$, we have $\sum_i a_i=N/2$. Hence there are exactly $N/2$ $i$'s where $a_i=1$. Denote such indices as $n_j$. Then Eq.~\eqref{eq.CFT.wavefunction} can be rewritten as $\psi_{\mathrm{CFT}}(n_1, ..., n_{N/2})\propto e^{i \pi \sum_i n_i} \prod_{n_i>n_j}^{N/2} (\sin(\pi(n_i-n_j)/N))^2$. This is precisely the HS wave function. See \cite{Cirac2010Infinite} for details. }. For convenience, when $m\neq n$, we denote
$\textstyle d_N\left(m,n\right) := |\sin\left(\pi\frac{m-n}{N}\right) |$. Eq.~\eqref{eq.CFT.wavefunction} becomes
\begin{eqnarray}\label{eq.CFT.spinchain}
\begin{split}
&\psi_{\mathrm{CFT}}(s_1,s_2,\ldots,s_{N}) \propto\\
&\delta\left( \sum_{i=1}^{N} s_i\right) 
\prod_{i ~\mathrm{odd}} s_i 
\prod_{m>n} d_N(m,n)^{\alpha s_m s_n}
\end{split}
\end{eqnarray}
The chiral boson CFT wave function Eq.~\eqref{eq.CFT.spinchain} is a special type of Jastrow wave function. Notice that a wave function $f(\{s_i\})$ is Jastrow if $f(\{s_i\})=\prod_{i<j}u_{ij}(s_i, s_j)$ for some function $u_{ij}(s_i,s_j)$.  Ref.~\cite{Glasser2018Neural} has shown that a generic Jastrow wave function can be exactly represented by a 1-layer RBM with $\mathcal{O}(N^2)$ hidden spins. In the following we first show that the chiral boson CFT correlator wave function  can be exactly represented by a 1-layer RBM with $\frac{N(N-1)}{2}+1$ hidden spins, 
\begin{equation}\label{eq.CFT.spinchain.toRBM.2}
\begin{split}
&\psi_{\mathrm{CFT}}(s_1,s_2,\ldots,s_{N}) \propto
e^{\ii\frac{\pi}{2} \sum_{i:\mathrm{odd}} (s_i-1)}\times  \\& \sum_{h=0}^{N-1} \sum_{\lbrace h_{m,n} \rbrace } e^{\ii\frac{\pi}{N} h \sum_{i} s_i - \ii \sum_{m>n} \left( s_m V_1^{\alpha,m,n} h_{m,n}+s_n V_2^{\alpha,m,n} h_{m,n} \right) }
\end{split}
\end{equation}
where the hidden spins include:
(a) $h_{m,n}=\pm 1$ coupling to the pair of visible spins $s_m$ and $s_n$ and
(b) $h=0,1,\ldots,N-1$ coupling to all visible spins, as a Zeeman magnetic field. In Eq.~\eqref{eq.CFT.spinchain.toRBM.2}, we have used $s_i= e^{i\pi (s_i-1)/2}$ for $s_i\in \{1, -1\}$. 
The hidden spin $h$ is introduced to impose the constraint $\sum_{i} s_i = 0$ by a $\mathbb{Z}_N$ Fourier transformation.
The total number of the hidden spins are $\frac{N(N-1)}{2}+1$. 
In Eq.~\eqref{eq.CFT.spinchain.toRBM.2} we use the notations:
\begin{eqnarray}
\begin{split}
& V_1^{\alpha,m,n}
:= V_1\left(\alpha \ln\left( d_N(m,n)\right)\right)	\\
& V_2^{\alpha,m,n} 
:= V_2\left(\alpha \ln\left( d_N(m,n)\right)\right)
\end{split}
\end{eqnarray}
where the functions $V_1(x)$ and $V_2(x)$ are:
\begin{eqnarray}\label{eq.V1V2.C=1/2}
\begin{split}
V_{1,2}(x)&= \frac{1}{2} \left(\text{arccosh}\left(e^x\right) \pm \text{arcsech}\left(e^x\right)\right).	\\
\end{split}
\end{eqnarray}
Fig.~\ref{fig.HS} is a graphical representation of the 1-layer RBM in Eq.~\eqref{eq.CFT.spinchain.toRBM.2}. The details to derive Eq.~\eqref{eq.CFT.spinchain.toRBM.2} can be found in section C of the supplementary material.

\begin{figure}[t]
	\centering
	\includegraphics[width=0.6\columnwidth]{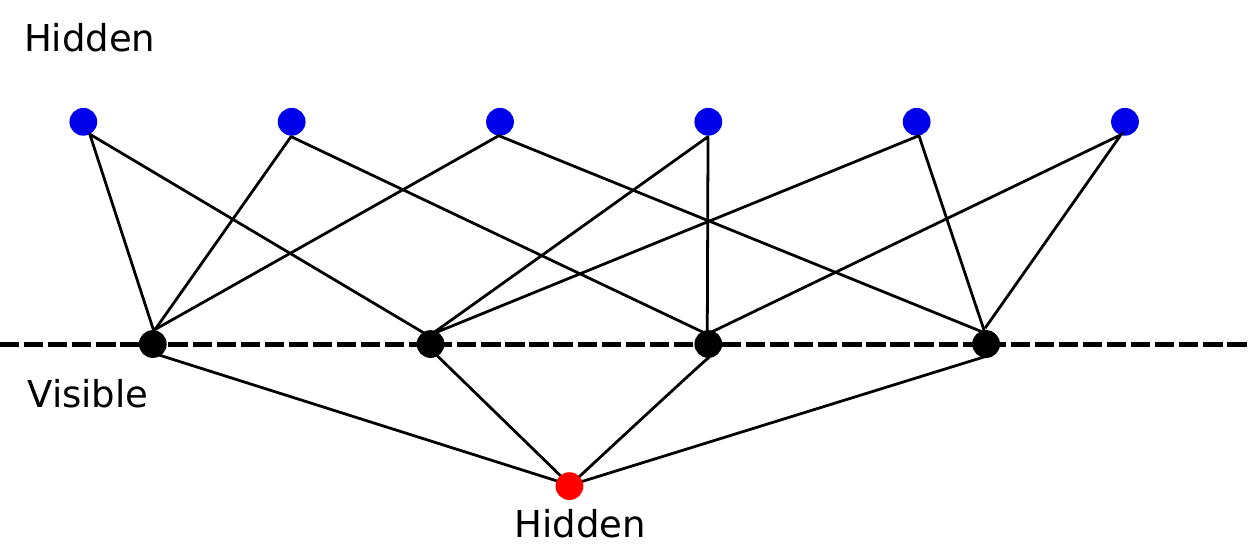}
	\caption{The RBM state figure for CFT correlator states with 4 spins with periodic boundary condition. The black dots are the visible spins and the blue dots are the hidden spins for each pair of visible spins, i.e., $h_{m,n}$. See Eq.~\eqref{eq.CFT.spinchain.toRBM.2}.}
	\label{fig.HS}
\end{figure}

We further numerically show that the CFT correlator wave function can be approximated by a 1-layer RBM with $N$ hidden spins (for $20\leq N\leq 22$) within the accuracy $99.9\%$ for $0<\alpha<0.1$, in contrast to the $\mathcal{O}(N^2)$ hidden spins for the exact representation Eq.~\eqref{eq.CFT.spinchain.toRBM.2}. We first prove that 1-layer RBM without bias, i.e., Eq.~\eqref{eq.RBM.1layer} where $a_i=0$ and $b_j=0$, can be approximated (with the precision given below) by an exponential function of a quadratic polynomial of physical spins, provided the weights are small, i.e., $|W_{ij}|\ll 1$. Concretely, 
\begin{equation}\label{eq.TargetState}
\begin{split}
\sum_{h_{1},\ldots,h_{M}} 
&\exp\left( \sum_{i,j} W_{ik} s_{i}h_{k} \right) 
\simeq 2^M \exp\left( \frac{1}{2}\sum_{i \neq j} U_{ij} s_i s_j \right) .
\end{split}
\end{equation}
where $U$ is symmetric and 
\begin{equation}\label{eq.alphadecompose}
U_{ij}\simeq \sum_{k=1}^M W_{ik}W_{jk} \quad \text{or}\quad U \simeq W W^T. 
\end{equation}

To prove Eq.~\eqref{eq.TargetState}, we use the \textit{cumulant expansion}. In Eq.~\eqref{eq.TargetState}, summing over the hidden spins amounts to calculating the expectation value with a uniform normalized probability distribution:
\begin{eqnarray}\label{eq.ProbabilityDistribution}
p\left( h_1,h_2,\ldots,h_M \right) = \frac{1}{2^{M}}, \quad \forall \; h_1,h_2,\ldots,h_M = \pm 1.
\end{eqnarray}
Denote $X:=\sum_{ij} W_{ij} s_i h_j$. 
Cumulant expansion\cite{McCullagh:2009} yields
\begin{eqnarray}\label{eq.cumulantexpansion}
\sum_{h_1,h_2,\ldots,h_{M}} e^{X} = 2^{M} \mathbb{E}(e^{X}) = 2^M \exp\left( \sum_{r=1}^{\infty} \frac{\kappa_r}{r!} \right)
\end{eqnarray}
where $\mathbb{E}$ denotes the expectation value over the probability distribution Eq.~\eqref{eq.ProbabilityDistribution}. 
The $r$-th cumulant $\kappa_r$ is listed in Ref.~\cite{McCullagh:2009}. $\kappa_r$ is of order $\mathcal{O}(|W_{ij}|^r)$. The first two cumulants $\kappa_1$ and $\kappa_2$ are: 
\begin{equation}
\begin{split}
&\kappa_1 =\mathbb{E}(X) = 0	\\
&\kappa_2 =\mathbb{E}\left(X^2 \right) = \sum_{i,j,k} W_{ik}W_{jk} s_i s_j
\end{split}
\end{equation}
Hence $\kappa_2$ is the leading contribution to Eq.~\eqref{eq.cumulantexpansion}. 
Keeping only the leading terms in Eq.~\eqref{eq.cumulantexpansion}, we obtain the desired approximation 
\begin{eqnarray}\label{Eq.1}
\begin{split}
\sum_{h_1,h_2,\ldots,h_{M}} e^{X} \simeq 2^M \exp\left( \frac{1}{2}\sum_{i,j,k} W_{ik}W_{jk} s_i s_j \right)
\end{split}
\end{eqnarray}
The right hand side of Eq.~\eqref{Eq.1} is precisely of the form $\exp\left( \frac{1}{2}\sum_{i \neq j} U_{ij} s_i s_j \right)$ where $U_{ij} = \sum_{k} W_{ik}W_{jk}$, and hence Eq.~\eqref{eq.alphadecompose} holds. Furthermore, since the only approximation comes from truncation of the Cumulant series in Eq.~\eqref{eq.cumulantexpansion}, we can determine that Eq.~\eqref{eq.alphadecompose} is accurate up to $\mathcal{O}(|W_{ij}|^3)$. 
Equivalently, the two sides of  Eq.~\eqref{eq.TargetState} equal up to a multiplicative factor $e^{\mathcal{O}(|W_{ij}|^3)}$ which is exponentially close to 1. We emphasize that the accuracy of the approximation also depends on the number of hidden spins $M$ which we specify below. 

\begin{figure}[t]
\centering
\hspace*{-1cm}
\includegraphics[width=1.2\columnwidth]{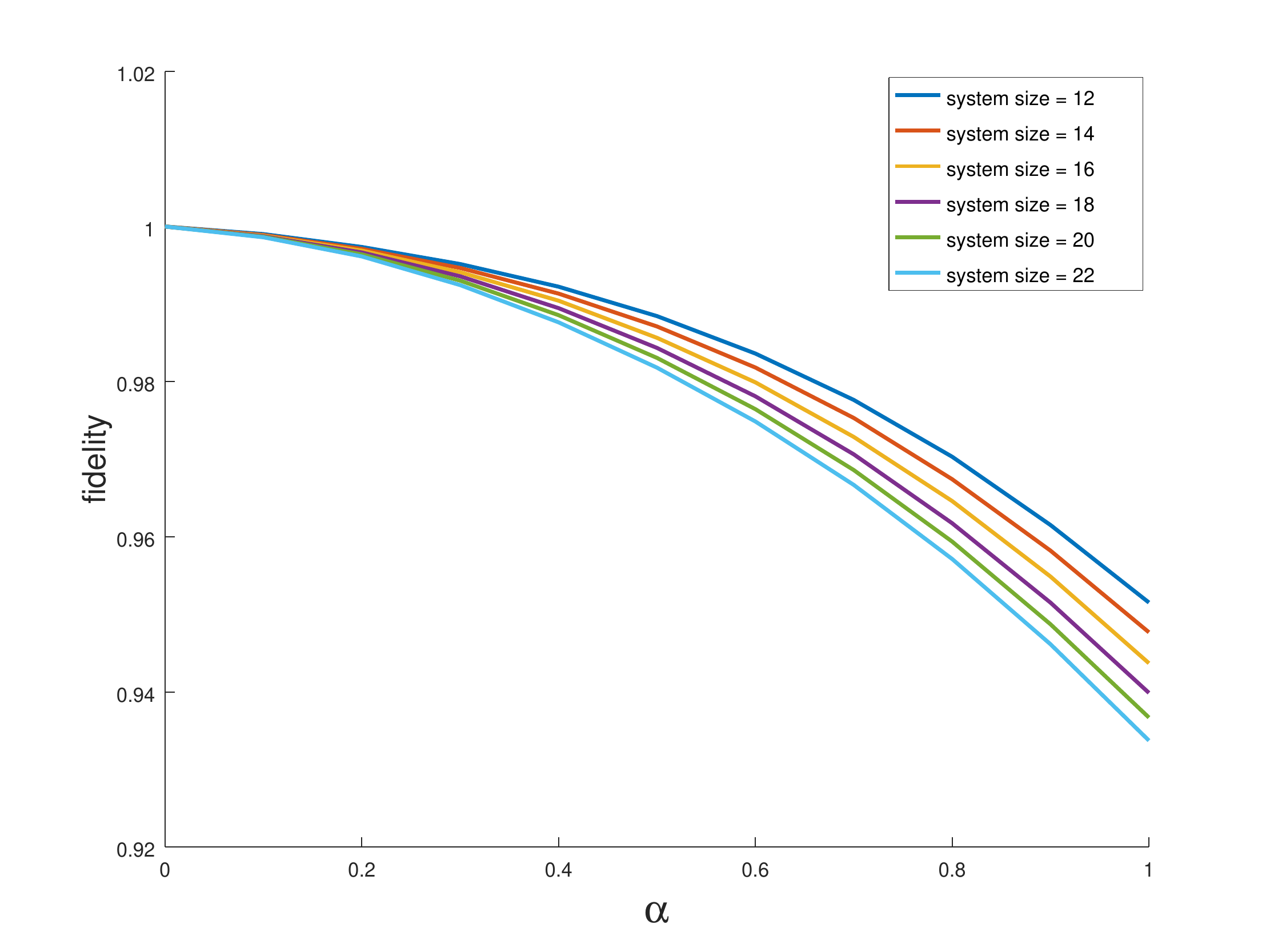}
\caption{Numerical results for the fidelity $|\overlap{\psi_{\mathrm{CFT}}}{\psi_{\mathrm{CFT-RBM}}}|^{1/N}$ versus $\alpha$. The two states $\psi_{\mathrm{CFT}}$ and $\psi_{\mathrm{CFT-RBM}}$ are defined in Eqs.~\eqref{eq.CFT.wavefunction} and \eqref{eq.NumericalTrialStates} respectively. Different curves corresponds to different system sizes $N$. The horizontal axis is $\alpha$, and the vertical axis is the fidelity $|\overlap{\psi_{\mathrm{CFT}}}{\psi_{\mathrm{CFT-RBM}}}|^{1/N}$. The fidelity is significantly larger than $99\%$ for $\alpha<0.1$.}
\label{fig.Overlap}
\end{figure}

Eq.~\eqref{eq.TargetState} provides a convenient approximation of the CFT correlator wave function using the RBM. For the simplicity of numerical simulations, we choose the number of hidden spins $M$ to be the number of visible spins $N$ (in contrast to $\mathcal{O}(N^2)$ for the exact representation Eq.~\eqref{eq.CFT.spinchain.toRBM.2}),  we numerically test the validity of this approximation by computing the overlap of $\psi_{\mathrm{CFT}}$ in Eq.~\eqref{eq.CFT.spinchain} with the RBM wave function: 
\begin{eqnarray}\label{eq.NumericalTrialStates}
\begin{split}
&\psi_{\mathrm{CFT-RBM}}(s_1,s_2,\ldots,s_{N}) \\
=& \frac{1}{Z}\delta\left( \sum_{i=1}^{N} s_i\right) 
e^{\ii\frac{\pi}{2} \sum_{i:\mathrm{odd}} (s_i-1) } \sum_{h_1,h_2,\ldots,h_{N}}
e^{\sum_{ij} W_{ij} s_i h_j}.
\end{split}
\end{eqnarray}
In Eq.~\eqref{eq.NumericalTrialStates}, $W$ is obtained by eigenvalue decomposing $U$ (from Eq.~\eqref{eq.CFT.wavefunction}) whose nonzero matrix elements are $\textstyle U_{i,j} = \alpha \ln\left( d_N(i,j)\right), \forall i\neq j$. Concretely, we decompose $U$ as $U=V \Lambda V^T$ where $V$ is real and $\lambda$ is the diagonal matrix, and we construct $W= V \sqrt{\Lambda}$ \footnote{Because $U_{i,j}$ is real and symmetric, there exists a real matrix $V$ satisfying $V V^T=V^T V=I$ such that $U=V \Lambda V^T$, where $\Lambda$ is a diagonal matrix. Denote $W= V \sqrt{\Lambda}$. $U$ can be rewritten as $U=W W^T$. Thus one can find $W$ by performing the eigenvalue decomposition $U=V \Lambda V^T$. }. We numerically compute the \emph{fidelity} (i.e. single-site overlap) $|\overlap{\psi_{\mathrm{CFT}}}{\psi_{\mathrm{CFT-RBM}}}|^{1/N}$. The numerical results are summarized in Fig.~\ref{fig.Overlap}. For $0\le \alpha \le 1$, the fidelity is greater than 90\% with sizes up to $N=22$. In particular, in the regime where the perturbation is valid, i.e., $0\leq \alpha \leq0.1$, the fidelity is greater than 99.9\%. This indicates that the approximation by cumulant expansion is acceptable at the leading order. The next order term in the cumulant expansion is $\kappa_4$ since $\kappa_3=0$ and it leads to a quadratic term in the wave function, that is $\sum_{i,j,k,l}U_{ijkl} s_i s_j s_k s_l$, which does not have the quadratic dependence of Eq.~\eqref{eq.TargetState}. 

We comment that the RBM wave function Eq.~\eqref{eq.NumericalTrialStates} also approximates the ground state wave function of the integrable XXZ chain. The Hamiltonian of the XXZ chain is 
\begin{equation}\label{Eq.XXZ}
    H_{XXZ}=\sum_{i=1}^N S_i^x S_{i+1}^x + S_i^y S_{i+1}^y + \Delta S_i^z S_{i+1}^z
\end{equation}
where $S_i^{a}$, $a=x,y,z$ are Pauli matrices. 
In Ref.~\cite{Cirac2010Infinite}, it was numerically shown that the overlap between the ground state of Eq.~\eqref{Eq.XXZ} $|\psi_{XXZ}\rangle$ and the chiral boson CFT correlator Eq.~\eqref{eq.CFT.correlator} $|\psi_{\mathrm{CFT}}\rangle$ is greater than $99\%$ with $N=20$ provided $\Delta= -\cos(2\pi \alpha)$ for $0\leq \alpha\leq 0.1$, i.e., $|\langle \psi_{XXZ}|\psi_{\mathrm{CFT}}\rangle|>99\%$. Thus the single site overlap (i.e. the fidelity) is
\begin{equation}
    |\langle \psi_{XXZ}|\psi_{\mathrm{CFT}}\rangle|^{1/N}>99.9\%, ~~~ N=20.
\end{equation}
Since we have already shown that the fidelity between the CFT correlator Eq.~\eqref{eq.CFT.correlator} and the CFT-RBM Eq.~\eqref{eq.NumericalTrialStates} for $N=20$ is also greater than $99.9\%$ for $0\leq \alpha\leq 0.1$, we conclude  that the CFT-RBM wave function nicely approximates the ground state of the XXZ ground state wave function with the fidelity \footnote{Since $|\langle \psi_{XXZ}|\psi_{\mathrm{CFT}}\rangle|^{1/N}>1-\epsilon$ where $\epsilon=0.1\%$, we find to the leading order of $\epsilon$ that the XXZ state can be approximated as $|\psi_{XXZ}\rangle= |\psi_{\mathrm{CFT}}\rangle+ N\epsilon |\alpha\rangle$ for some $|\alpha\rangle$. Similarly since $|\overlap{\psi_{\mathrm{CFT}}}{\psi_{\mathrm{CFT-RBM}}}|^{1/N}> 1-\epsilon$, to the leading order of $\epsilon$, the $|\psi_{\mathrm{CFT-RBM}}\rangle$  can be approximated as $|\psi_{\mathrm{CFT-RBM}}\rangle= |\psi_{\mathrm{CFT}}\rangle + N\epsilon |\beta\rangle$ for some $|\beta\rangle$. We estimate the fidelity $|\langle \psi_{XXZ}|\psi_{\mathrm{CFT-RBM}}\rangle|^{1/N}$  up to leading order of $\epsilon$ to be $(1+ N \epsilon (\langle \alpha |\mathrm{CFT}\rangle+ \langle \mathrm{CFT}|\beta\rangle))^{1/N}\gtrsim 1-2\epsilon=99.8\%$. }
\begin{equation}
    |\langle \psi_{XXZ}|\psi_{\mathrm{CFT-RBM}}\rangle|^{1/N}\gtrsim 99.8\%, ~~~ N=20.
\end{equation}

We finally comment that Eq.~\eqref{eq.TargetState} allows us to  label  the orthogonal basis of the chiral boson CFT using hidden spins. Recalling Eq.~\eqref{eq.CFT.correlator} and \eqref{eq.CFT.wavefunction}, the matrix elements $U_{ij}$ in Eq.~\eqref{eq.TargetState} can be identified as:
\begin{eqnarray}\label{Eq.2006}
U_{mn} = \alpha \ln\sin\left( \frac{\pi(m-n)}{N}\right), \quad m>n.
\end{eqnarray}
$U_{mn}$ in Eq.~\eqref{Eq.2006} is proportional to the chiral boson correlation function on a cylinder of length $2\pi$:
\begin{eqnarray}
U_{mn} \propto \alpha \langle \phi\left(\frac{2\pi m}{N}\right)\phi\left(\frac{2\pi n}{N}\right) \rangle_{\mathrm{cyl}}
\end{eqnarray}
Denote $\ket{p,q}$ as the complete and orthonormal basis of chiral boson CFT, consisting of all the primary states (labeled by $p$) and their descendents (labeled by $q\geq 0$).  Inserting the identity $1\equiv \sum_{p,q} \ket{p,q}\bra{p,q}$ on the cylinder, 
we have
\begin{eqnarray}\label{eq.CompleteBasisInsertion}
U_{mn} \propto \alpha \sum_{p,q} \langle \phi\left(\frac{2\pi m}{N}\right)\ket{p,q}_{\mathrm{cyl}} \bra{p,q}\phi\left(\frac{2\pi n}{N}\right) \rangle_{\mathrm{cyl}}
\end{eqnarray}
Comparing Eq.~\eqref{eq.CompleteBasisInsertion} and \eqref{eq.alphadecompose}, we can find that:
\begin{eqnarray}\label{eq.hiddenSpinCFTBasis}
W_{m(p,q)} \propto \sqrt{\alpha} \langle \phi\left( \frac{2\pi m}{N}\right) \ket{p,q}_{\mathrm{cyl}},
\end{eqnarray}
where the second index $(p,q)$ of $W$, i.e., the hidden spin in the 1-layer RBM, labels   the orthonormal basis of chiral boson CFT on a cylinder $|p,q\rangle$. Since there are infinite number of orthogonal CFT bases, the number of hidden spins $M$ is infinite in this case.

\paragraph*{Conclusion}

In this paper, we study two famous states for 1D spin chains: the AKLT state and the CFT correlator state. The AKLT state can be exactly represented as a 2-layer RBM state, where the number of hidden spins per visible spin is 5. We further proved that any MPS can be approximated by a 2-layer RBM within a given accuracy $\epsilon$. The free-boson CFT correlator wave function Eq.~\eqref{eq.CFT.wavefunction} can be exactly represented by a 1-layer RBM, and the number of hidden spins per unit cell is approximately $M=\mathcal{O}(\frac{N^2}{2})$ as $N$ becomes large. The number of hidden spins grows as the system size indicates that the ground state is strongly entangled, with the entanglement bounded from above by $\log 2^{NM}$. Moreover, using cumulant expansion for the 1-layer RBM, the CFT correlator wave function can be approximated well by a 1-layer RBM with $N$ number of hidden spins with the accuracy shown in  Fig.~\ref{fig.Overlap}. It is worth emphasizing that the hidden spins can be interpreted as the orthonormal basis of the free boson CFT, as shown in Eq.~\eqref{eq.hiddenSpinCFTBasis}.

\paragraph*{Acknowledgment}
We thank Nicolas Regnault for useful discussions and the collaborations in the previous works. 
B. A. B. was supported by the Guggenheim Foundation, the Department of Energy Grant No. DE-SC0016239, the National Science Foundation EAGER Grant No. DMR 1643312, Simons Investigator Grants No. 404513, ONR No. N00014-14-1-0330, and  NSF-MRSEC No. DMR-142051, the Packard Foundation, the Schmidt Fund for Innovative Research. G. S. acknowledges the  support from the grants PGC2018-095862-B-C21, QUITEMAD+ S2013/ICE-2801, SEV-2016-0597 of the``Centro de Excelencia Severo Ochoa” Programme and the CSIC Research Platform on Quantum Technologies PTI-001.

\bibliography{RBM_ExoticState}

\clearpage
\onecolumngrid
\appendix

\setcounter{equation}{0}
\renewcommand{\theequation}{A\thechapter.\arabic{equation}}
\renewcommand{\thetheorem}{A\thechapter.\arabic{theorem}}

\section{A. Expressing the CG Coefficients as an RBM}
\label{App.CG}

In this appendix, we express the Clebsch-Gordan (CG) coefficients in Eq.~\eqref{eq.project_tensor} in the main text in terms of a 1-layer RBM. We find that the coefficients can be organized into the following formula,
\begin{eqnarray}\label{eq.project_tensor.tospin1.0}
\begin{split}
P^{s}_{a,b} 
&= \frac{1}{3} \sum_{h=-1}^{1} e^{\frac{2\pi\ii}{3} h (a+b-s)}F_{a,b}	\\
\end{split}
\end{eqnarray}
where
\begin{eqnarray}
F_{\frac{1}{2},\frac{1}{2}}=F_{-\frac{1}{2},-\frac{1}{2}}=1,	\quad
F_{\frac{1}{2},-\frac{1}{2}}=F_{-\frac{1}{2},\frac{1}{2}}=\frac{1}{\sqrt{2}}
\end{eqnarray} 
The rest of the task is to write the factor $F_{a,b}$ as an RBM. Fortunately, $F_{a,b}$ fits nicely into the formula
\begin{eqnarray}\label{eq.project_tensor.tospin1.2}
\begin{split}
F_{a,b} 
=& \cos\left(\frac{\pi}{4}(a-b)\right)
= \frac{1}{2} \sum_{h^\prime=\pm \frac{1}{2}} e^{\ii \frac{\pi}{2}h^\prime(a-b) } \\
\end{split}
\end{eqnarray}
Combining both Eqs.~\eqref{eq.project_tensor.tospin1.0} and \eqref{eq.project_tensor.tospin1.2} leads to:
\begin{eqnarray}\label{eq.projector.as.BM1}
\begin{split}
P^{s}_{a,b}
= \frac{1}{6} \sum_{h=-1}^{1}\sum_{h^\prime=\pm \frac{1}{2}} e^{\ii\left( \frac{2\pi}{3} h (a+b-s)+ \frac{\pi}{2}h^\prime(a-b) \right)} \\
\end{split}
\end{eqnarray}

\setcounter{equation}{0}
\renewcommand{\theequation}{B\thechapter.\arabic{equation}}
\renewcommand{\thetheorem}{B\thechapter.\arabic{theorem}}

\section{B. Expressing a finitely connected and translationally invariant RBM state as an MPS}

In this appendix, we map a finitely connected and translationally invariant 2-layer RBM into an MPS. See Ref.~\cite{zheng2018restricted} for the discussion of 1-layer RBM $\leftrightarrow$ MPS map. As an example, we explain the 2-layer RBM $\leftrightarrow$ MPS map where there are 1 visible spin, 2 hidden spins in the first hidden layer and 2 hidden spins in the second hidden layer per unit cell. We will comment on the generalization to 2-layer RBM's with more visible and hidden spins per unit cell. 

Consider a translational invariant and nearest neighbor connected 2-layer RBM defined on a 1D lattice with $N$ unit cells. We first introduce the degrees of freedom in each unit cell as well as the weights and biases. 
\begin{enumerate}
\item There is 1 visible spin $s_r$ in the $r$-th unit cell, $r=1, ..., N$. 
\item There are 2 hidden spins $h^1_r$ and $g^1_r$ in the $r$-th unit cell of the first hidden layer. $h^1_r$ only connects to the visible spin in the same unit cell $s_r$, while $g^1_r$ connects to the visible spins belonging to neighbor unit cells $s_r$ and $s_{r+1}$. 
\item There are 2 hidden spins $h^2_r$ and $g^2_r$ in the $r$-th unit cell of the second hidden layer. $h^2_r$ only connects to the hidden spins from the first hidden layer in the same unit cell, i.e., $h^1_r$ and $g^1_r$. $g^2_r$ connects to the hidden spins from the first hidden layer belonging to neighbor unit cells, i.e., $h^1_r, g^1_r, h^1_{r+1}$ and $g^1_{r+1}$. 
\item The weights connecting the visible spins and hidden spins from the first hidden layer are specified as follows:
\begin{enumerate}
\item The weight connecting $s_r$ and $h^1_r$ is $W_{sh}^1$. 
\item The weight connecting $s_r$ and $g^1_r$ is $W_{sg}^1$. 
\item The weight connecting $s_{r+1}$ and $g^1_r$ is $\widetilde{W}_{sg}^1$. 
\end{enumerate}
\item The weights connecting the hidden spins from the first hidden layer and those from the second hidden layer are as follows:
\begin{enumerate}
\item The weight connecting $h^1_r$ and $h^2_r$ is $W_{hh}^2$. 
\item The weight connecting $g^1_r$ and $h^2_r$ is $W_{gh}^2$.
\item The weight connecting $h^1_r$ and $g^2_r$ is $W_{hg}^2$.
\item The weight connecting $g^1_r$ and $g^2_r$ is $W_{gg}^2$.
\item The weight connecting $h^1_{r+1}$ and $g^2_r$ is $\widetilde{W}_{hg}^2$.
\item The weight connecting $g^1_{r+1}$ and $g^2_r$ is $\widetilde{W}_{gg}^2$.   
\end{enumerate}
\item The biases are as follows:
\begin{enumerate}
\item The bias of $s_r$ is $A_r$. 
\item The bias of $h_r^i$ is $B_r^i$, $i=1, 2$. 
\item The bias of $g_r^i$ is $C_r^i$, $i=1,2$. 
\end{enumerate}
\end{enumerate}
For simplicity, we assume all the visible spins and hidden spins are $\{0,1\}$ valued. 
The 2-layer RBM takes the above general form
\begin{equation}\label{Eq.2layerRBM}
\begin{split}
f_{2-\mathrm{RBM}}(s_1, ..., s_N)= &\sum_{\{h^1_{r}\}, \{g^1_r\}}\sum_{\{h^2_{r}\}, \{g^2_r\}}\exp\Bigg( \sum_{r=1}^{N}\bigg(W_{sh}^1 s_r h_r^1 + W_{sg}^1 s_r g_r^1+ \widetilde{W}_{sg}^1 s_{r+1} g_r^1 \\
&+ W_{hh}^2 h_r^1 h_r^2 + W_{gh}^2 g_r^1 h_r^2 + W^2_{hg} h_r^1 g_r^2 + W^2_{gg} g_r^1 g_r^2+ \widetilde{W}_{hg}^2 h_{r+1}^1 g_r^2+ \widetilde{W}^2_{gg} g_{r+1}^1 g_r^2 \\
&+ A_rs_r + B_r^1 h_r^1 + C_r^1 g_r^1+ B_r^2 h_r^2+ C_r^2g_r^2 \bigg)  \Bigg)
\end{split}
\end{equation}
See Figure \ref{Fig.2layerRBM} for a graphical representation of the RBM. 

\begin{figure}[t]
\centering
\includegraphics[width=0.6\columnwidth]{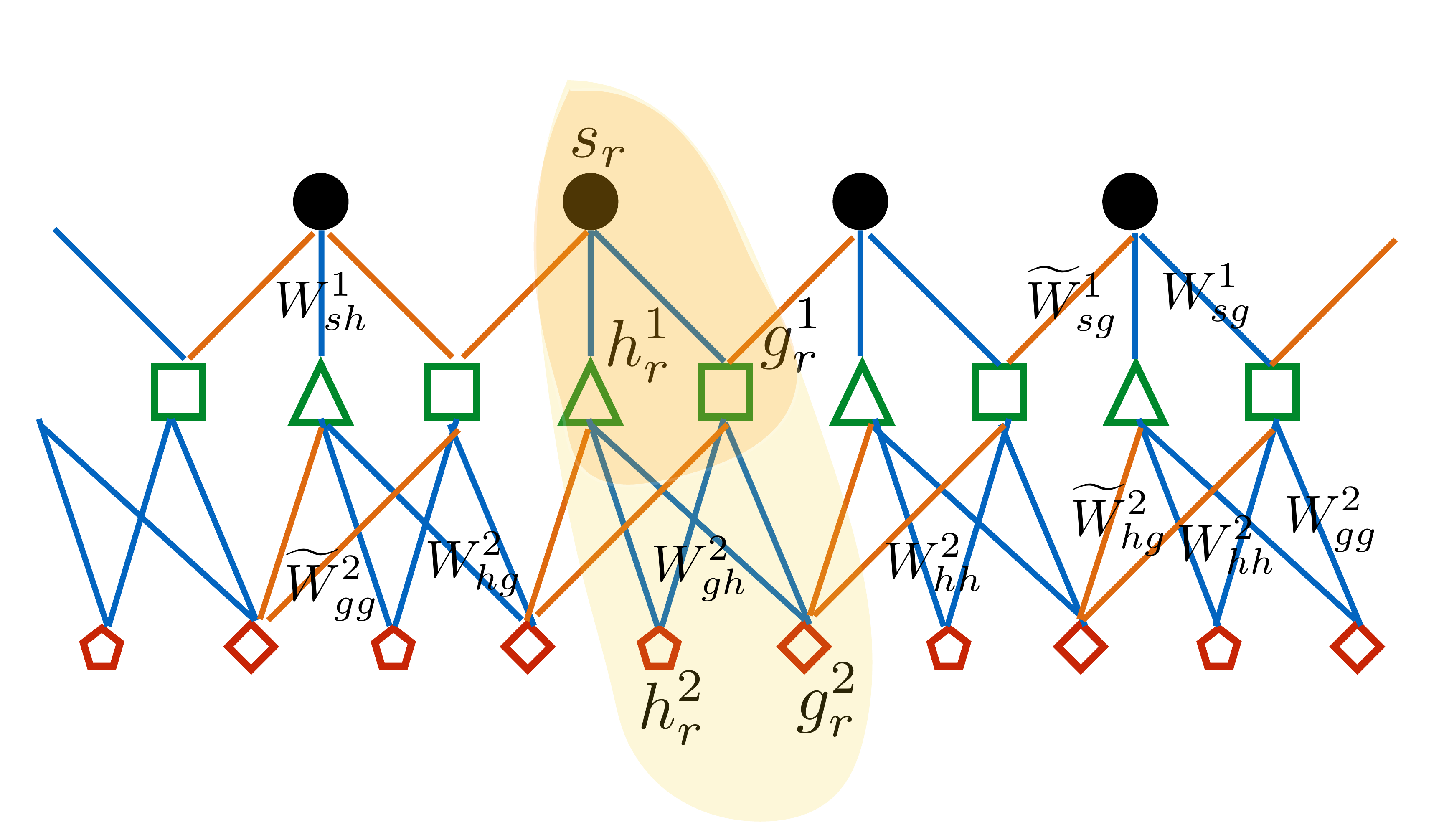}
\caption{Translationally invariant 2-layer RBM with nearest neighbor connection. The black dot represents the visible spin $\{s_r\}$. In the first hidden layer, the green triangle represents the hidden spin $\{h^1_r\}$, and the green box represents the hidden spin $\{g^1_r\}$.  In the second hidden layer, the red pentagon represents the hidden spin $\{h^2_r\}$, and the red diamond represents the hidden spin $\{g^2_r\}$. The blue lines represent the weights connecting the spins within the unit cell, and the orange lines represent the weights connecting the spins from adjacent unit cells (which are labeled by $\widetilde{W}$'s). The yellow shaded region represents a unit cell involving all three layers. The dark yellow shaded region represents a unit cell involving the visible layer and the first hidden layer.  }
\label{Fig.2layerRBM}
\end{figure}

We further organize the 2-layer RBM Eq.~\eqref{Eq.2layerRBM} as an MPS. We first construct a state by treating the RBM Eq.~\eqref{Eq.2layerRBM} as a wave function, 
\begin{eqnarray}
|\mathrm{RBM}\rangle= \sum_{\{s_r\}} f_{2-\mathrm{RBM}}(s_1, ..., s_N) |\{s_r\}\rangle
\end{eqnarray}
By properly grouping the various terms in Eq.~\eqref{Eq.2layerRBM}, we find 
\begin{eqnarray}
|\mathrm{RBM}\rangle= \sum_{\{s_r\}} \Tr \bigg(\prod_{r=1}^N T^{s_r}\bigg) |\{s_r\}\rangle
\end{eqnarray}
where the tensor $T^{s_r}$ is a $4\times 4$ matrix
\begin{eqnarray}\label{2RBM1}
\begin{split}
T^{s_r}_{g_{r-1}^1 g_{r-1}^2, g_r^1 g_{r}^2} = &\sum_{h_r^1, h_r^2} \exp\Bigg(W_{sh}^1 s_r h_r^1 + W_{sg}^1 s_r g_r^1+ \widetilde{W}_{sg}^1 s_{r} g_{r-1}^1 \\
&+ W_{hh}^2 h_r^1 h_r^2 + W_{gh}^2 g_r^1 h_r^2 + W^2_{hg} h_r^1 g_r^2 + W^2_{gg} g_r^1 g_r^2+ \widetilde{W}_{hg}^2 h_{r}^1 g_{r-1}^2+ \widetilde{W}^2_{gg} g_{r}^1 g_{r-1}^2 \\
&+ A_rs_r + B_r^1 h_r^1 + C_r^1 g_r^1+ B_r^2 h_r^2+ C_r^2g_r^2 \Bigg)
\end{split}
\end{eqnarray}
This shows that one can directly obtain the MPS from a 2-layer translationally invariant RBM. It is remarkable that by grouping the indices $g_{r-1}^1, g_{r-1}^2$ as the left index and the indices $g_r^1, g_{r}^2$ as the right index, the rank of the 4-by-4 matrices $T^{s_r}$ is upper bounded by $2$. 
To see this, we denote
\begin{eqnarray}\label{TUV}
T^{s_r}_{g_{r-1}^1 g_{r-1}^2, g_r^1 g_{r}^2} = \sum_{h_r^1} U^{s_r}_{g_{r-1}^1 g_r^1 h_r^1} V^{g_{r-1}^1 g_r^1 h_r^1}_{g_{r-1}^2 g_r^2}
\end{eqnarray}
where 
\begin{equation}
\begin{split}
U^{s_r}_{g_{r-1}^1 g_r^1 h_r^1}&= \exp\Bigg(W_{sh}^1 s_r h_r^1 + W_{sg}^1 s_r g_r^1+ \widetilde{W}_{sg}^1 s_{r} g_{r-1}^1+A_rs_r + B_r^1 h_r^1 + C_r^1 g_r^1 \Bigg)\\
V^{g_{r-1}^1 g_r^1 h_r^1}_{g_{r-1}^2 g_r^2}&= \sum_{h_r^2} \exp\Bigg(W_{hh}^2 h_r^1 h_r^2 + W_{gh}^2 g_r^1 h_r^2 + W^2_{hg} h_r^1 g_r^2 + W^2_{gg} g_r^1 g_r^2+ \widetilde{W}_{hg}^2 h_{r}^1 g_{r-1}^2+ \widetilde{W}^2_{gg} g_{r}^1 g_{r-1}^2+ B_r^2 h_r^2+ C_r^2g_r^2\Bigg)
\end{split}
\end{equation}
Notice that $U^{s_r}_{g_{r-1}^1 g_r^1 h_r^1}$ and $V^{g_{r-1}^1 g_r^1 h_r^1}_{g_{r-1}^2 g_r^2}$ are both 1-layer RBMs respectively. In Ref.~\cite{zheng2018restricted}, we showed that the MPS matrix of the 1-layer RBM must have rank 1. From Eq.~\eqref{TUV} the matrix $T$ can be expressed as sum of two rank 1 $4\times 4$ matrices, 
\begin{eqnarray}\label{B7}
T^{s_r}_{g_{r-1}^1 g_{r-1}^2, g_r^1 g_{r}^2} = U^{s_r}_{g_{r-1}^1 g_r^1 0} V^{g_{r-1}^1 g_r^1 0}_{g_{r-1}^2 g_r^2} + U^{s_r}_{g_{r-1}^1 g_r^1 1} V^{g_{r-1}^1 g_r^1 1}_{g_{r-1}^2 g_r^2} 
\end{eqnarray}
because $U^{s_r}_{g_{r-1}^1 g_r^1 0} V^{g_{r-1}^1 g_r^1 0}_{g_{r-1}^2 g_r^2} $ and $U^{s_r}_{g_{r-1}^1 g_r^1 1} V^{g_{r-1}^1 g_r^1 1}_{g_{r-1}^2 g_r^2} $ are both rank 1 $4\times 4$ matrices. (Tensor product of two rank 1 matrices is still a rank 1 matrix. )
Hence the rank of $T$ is at most 2. Actually the rank of $T$ can be 1, when $U^{s_r}_{g_{r-1}^1 g_r^1 h_r^1} V^{g_{r-1}^1 g_r^1 h_r^1}_{g_{r-1}^2 g_r^2} $ is independent of $h_r^1$. In particular, when all the weights connecting to the hidden spin $h_r^1$ vanish, $U^{s_r}_{g_{r-1}^1 g_r^1 h_r^1} V^{g_{r-1}^1 g_r^1 h_r^1}_{g_{r-1}^2 g_r^2} $ is independent of $h_r^1$. We realize that even in this case (where all weights connecting to $h_r^1$ vanish), the Boltzmann machine is still a 2-layer RBM (which does not reduce to a 1-layer RBM). In conclusion, we showed the upper bound for the MPS of the 2-layer RBM Eq.~\eqref{2RBM1} is 2.

Generalization to the RBM with more visible and hidden spins per unit cell is straightforward. The recipe is to introduce species indices to label visible and hidden spins, i.e., $\{s_r^a, h^{1, b}_{r}, g^{1, c}_{r}, h^{2, e}_{r}, g^{2, f}_{r}\}$ where $a, b, c, e, f$ are indices labeling the species. More specifically,
\begin{enumerate}
	\item $s^a_r$ is the visible spin. $a=1, ..., n_s$. 
\item $h^{1, b}_{r}$, $b=1, ..., n_{h1}$, is the hidden spin from the first hidden layer connecting only with the visible spins from the same unit cell, i.e., $s_{r}^a$. 
\item $g^{1, c}_{r}$, $c=1,..., n_{g1}$, is the hidden spin from the first hidden layer connecting only with the visible spins from the neighbor unit cells, i.e., $s_{r}^a$ and $s_{r+1}^a$; 
\item $h^{2, e}_{r}$, $e=1, ..., n_{h2}$,  is the hidden spin from the second hidden layer connecting only with the hidden spins from the first hidden layer from the same unit cell, i.e., $h_{r}^{1,b}$ and $g_r^{1,c}$;
\item $g^{2, f}_{r}$, $f=1, ..., n_{g2}$ is the hidden spin from the second hidden layer connecting only with the hidden spins from the first hidden layer from adjacent unit cells, i.e., $h^{1, b}_{r}, g^{1, c}_{r}, h^{1, b}_{r+1}$ and $g^{1, c}_{r+1}$.
\end{enumerate}
Finally, the weights and biases $\{W^1_{sh}, W_{sg}^1, \widetilde{W}_{sg}^1, W_{hh}^2, W_{gh}^2, W_{hg}^2, W_{gg}^2, \widetilde{W}_{hg}^2, \widetilde{W}_{gg}^2, A_r, B_r^1, C_r^1, B_r^2, C_r^2\}$ are generalized with the species labels: $\{W^{1, ab}_{sh}, W_{sg}^{1,ac}, \widetilde{W}_{sg}^{1,ac}, W_{hh}^{2,be}, W_{gh}^{2,ce}, W_{hg}^{2,bf}, W_{gg}^{2,cf}, \widetilde{W}_{hg}^{2,bf}, \widetilde{W}_{gg}^{2,cf}, A_r^a, B_r^{1,b}, C_r^{1,c}, B_r^{2,e}, C_r^{2,f}\}$. Using the above notations, the 2-layer RBM is
\begin{equation}\label{Eq. 102}
\begin{split}
f_{2-\mathrm{RBM}}(\{s^a_r\})= &\sum_{\{h^{1,b}_{r}\}, \{g^{1,c}_r\}}\sum_{\{h^{2,e}_{r}\}, \{g^{2,f}_r\}}\exp\Bigg( \sum_{r=1}^{N}\bigg(W_{sh}^{1,ab} s^a_r h_r^{1,b} + W_{sg}^{1,ac} s^a_r g_r^{1, c}+ \widetilde{W}_{sg}^{1,ac} s^a_{r+1} g_r^{1,c} \\
&+ W_{hh}^{2,be} h_r^{1,b} h_r^{2,e} + W_{gh}^{2,ce} g_r^{1,c} h_r^{2,e} + W^{2,bf}_{hg} h_r^{1,b} g_r^{2,f} + W^{2,cf}_{gg} g_r^{1,c} g_r^{2,f}+ \widetilde{W}_{hg}^{2,bf} h_{r+1}^{1,b} g_r^{2,f}+ \widetilde{W}^{2,cf}_{gg} g_{r+1}^{1,c} g_r^{2,f} \\
&+ A^a_rs^a_r + B_r^{1,b} h_r^{1,b} + C_r^{1,c} g_r^{1,c}+ B_r^{2,e} h_r^{2,e}+ C_r^{2,f}g_r^{2,f} \bigg)  \Bigg)
\end{split}
\end{equation}
The RBM state of the 2-layer RBM Eq.~\eqref{Eq. 102} can be rewritten as an MPS, 
\begin{eqnarray}
|\mathrm{RBM}\rangle=\sum_{\{s^a_r\}} f_{2-\mathrm{RBM}}(\{s^a_r\}) |\{s^a_r\}\rangle= \sum_{\{s^a_r\}} \Tr \bigg(\prod_{r=1}^N T^{s^1_r... s^{n_s}_r}\bigg) |\{s_r\}\rangle
\end{eqnarray}
where 
\begin{equation}\label{2RBM2}
\begin{split}
T^{\{s^a_r\}}_{\{g_{r-1}^{1,c}\} \{g_{r-1}^{2,f}\}, \{g_r^{1,c'}\} \{g_{r}^{2,f'}\}} = &\sum_{\{h_r^{1,b}\}, \{h_r^{2,e}\}} \exp\Bigg(W_{sh}^{1,ab} s^a_r h_r^{1,b} + W_{sg}^{1,ac'} s^a_r g_r^{1,c'}+ \widetilde{W}_{sg}^{1,ac} s^a_{r} g_{r-1}^{1,c} \\
&+ W_{hh}^{2,be} h_r^{1,b} h_r^{2,e} + W_{gh}^{2,c'e} g_r^{1,c'} h_r^{2,e} + W^{2,bf'}_{hg} h_r^{1,b} g_r^{2,f'} + W^{2,c'f'}_{gg} g_r^{1,c'} g_r^{2,f'}+ \widetilde{W}_{hg}^{2,bf} h_{r}^{1,b} g_{r-1}^{2,f}\\&+ \widetilde{W}^{2,c'f}_{gg} g_{r}^{1,c'} g_{r-1}^{2,f} 
+ A^a_rs^a_r + B_r^{1,b} h_r^{1,b} + C_r^{1,c'} g_r^{1,c'}+ B_r^{2,e} h_r^{2,e}+ C_r^{2,f'}g_r^{2,f'} \Bigg)
\end{split}
\end{equation}
which is an $2^{n_{g1}+n_{g2}}\times 2^{n_{g1}+n_{g2}}$ dimensional matrix. We further calculate the rank of the MPS Eq.~\eqref{2RBM2}. We use the same method as in Eq.~\eqref{B7}. We denote 
\begin{eqnarray}
T^{\{s^a_r\}}_{\{g_{r-1}^{1,c}\} \{g_{r-1}^{2,f}\}, \{g_r^{1,c'}\} \{g_{r}^{2,f'}\}} = \sum_{\{h_r^{1,b}\}} U^{\{s_r^a\}}_{\{g_{r-1}^{1,c}\} \{g_r^{1,c'}\} \{h_r^{1,b}\}} V^{\{g_{r-1}^{1,c}\} \{g_r^{1,c'}\} \{h_r^{1,b}\}}_{\{g_{r-1}^{2,f}\} \{g_r^{2,f'}\}}
\end{eqnarray}
where 
\begin{eqnarray}
\begin{split}
U^{\{s_r^a\}}_{\{g_{r-1}^{1,c}\} \{g_r^{1,c'}\} \{h_r^{1,b}\}} &= \exp\Bigg(W_{sh}^{1,ab} s^a_r h_r^{1,b} + W_{sg}^{1,ac'} s^a_r g_r^{1,c'}+ \widetilde{W}_{sg}^{1,ac} s^a_{r} g_{r-1}^{1,c}  + A^a_rs^a_r + B_r^{1,b} h_r^{1,b} + C_r^{1,c'} g_r^{1,c'} \Bigg)\\
V^{\{g_{r-1}^{1,c}\} \{g_r^{1,c'}\} \{h_r^{1,b}\}}_{\{g_{r-1}^{2,f}\} \{g_r^{2,f'}\}} &= \sum_{ \{h_r^{2,e}\}} \exp\Bigg(W_{hh}^{2,be} h_r^{1,b} h_r^{2,e} + W_{gh}^{2,c'e} g_r^{1,c'} h_r^{2,e} + W^{2,bf'}_{hg} h_r^{1,b} g_r^{2,f'} + W^{2,c'f'}_{gg} g_r^{1,c'} g_r^{2,f'}+ \widetilde{W}_{hg}^{2,bf} h_{r}^{1,b} g_{r-1}^{2,f}\\&+  \widetilde{W}^{2,c'f}_{gg} g_{r}^{1,c'} g_{r-1}^{2,f} +  B_r^{2,e} h_r^{2,e}+ C_r^{2,f'}g_r^{2,f'}   \Bigg)
\end{split}
\end{eqnarray}
Both $U^{\{s_r^a\}}_{\{g_{r-1}^{1,c}\} \{g_r^{1,c'}\} \{h_r^{1,b}\}} $ and $V^{\{g_{r-1}^{1,c}\} \{g_r^{1,c'}\} \{h_r^{1,b}\}}_{\{g_{r-1}^{2,f}\} \{g_r^{2,f'}\}} $ are rank 1 matrices. Therefore, we find that $T$ is a sum of $2^{n_{h1}}$ rank 1 matrices. In summary we find that  the $2^{n_{g1}+n_{g2}}\times 2^{n_{g1}+n_{g2}}$ matrix $T^{\{s^a_r\}}_{\{g_{r-1}^{1,c}\} \{g_{r-1}^{2,f}\}, \{g_r^{1,c'}\} \{g_{r}^{2,f'}\}}$ has maximal rank $\min\{2^{n_{h1}}, 2^{n_{g1}+n_{g2}}\}$. 

\setcounter{equation}{0}
\renewcommand{\theequation}{C\thechapter.\arabic{equation}}
\renewcommand{\thetheorem}{C\thechapter.\arabic{theorem}}

\section{C. Exact RBM Representation for CFT Correlation Wave Function}
\label{App.C}

For the calculation convenience, we introduce the following identity:
\begin{lemma}\label{lemma.ssV2RBM}
\begin{equation}\label{eq.ssV2RBM}
\exp\left( s_1s_2 V\right) = \frac{1}{2} \sum_{h=\pm 1} \exp\left( V_1 s_1 h + V_2 s_2 h\right), \quad
\forall\; s_1,s_2=\pm 1,
\end{equation}
where $V$ is a constant, and $V_1, V_2$ will be shown to only depend on $V$.
\end{lemma}

\begin{proof}
Summing over the $h$ on RHS of Eq.~\eqref{eq.ssV2RBM} leads to:
\begin{equation}
\begin{split}
\exp\left( s_1s_2 V\right) = \cosh\left( V_1 s_1 + V_2 s_2\right), \quad
\forall\; s_1,s_2=\pm 1.
\end{split}
\end{equation}
More explicitly, this set of equations without duplicates is listed below:
\begin{equation}
\begin{split}
e^V&=\cosh (V_1+V_2),	\\
e^{-V}&= \cosh (V_1-V_2)
\end{split}
\end{equation}
We can solve these two equations to find a solution for $V_1$ and $V_2$:
\begin{equation}\label{eq.V1V2}
\begin{split}
V_1&= \frac{1}{2} \left(\mathrm{sech}^{-1}\left( e^{-V}\right)+\mathrm{sech}^{-1}\left( e^V\right)\right),	\\
V_2&= \frac{1}{2} \left(\mathrm{sech}^{-1}\left( e^{-V}\right)-\mathrm{sech}^{-1}\left( e^V\right)\right).	\\
\end{split}
\end{equation}
We find that $V_1$ and $V_2$ are Eq.~\eqref{eq.V1V2.C=1/2} in the main text. 

\end{proof}

We are ready to transform the CFT correlator wave function in Eq.~\eqref{eq.CFT.spinchain} to an RBM state. Our main tool is the identity  Eq.~\eqref{eq.ssV2RBM}. To begin with, let us apply Eq.~\eqref{eq.ssV2RBM} to the following term: 
\begin{equation}\label{eq.ssV2RBM.oneterm}
\begin{split}
&d\left( \frac{m-n}{N} \right)^{\alpha s_m s_n}	\\
=& \exp\left(\ln\left( d\left(\frac{m-n}{N}\right)\right) \alpha s_ms_n\right)	\\
=& \frac{1}{2} \sum_{h_{m,n}=\pm 1} \exp\bigg( s_m V_1\left(\alpha \ln\left( d\left( \frac{m-n}{N}\right)\right)\right) h_{m,n} +s_n V_2\left(\alpha \ln\left( d\left( \frac{m-n}{N}\right)\right)\right) h_{m,n}\bigg)
\end{split}
\end{equation}
where Eq.~\eqref{eq.ssV2RBM} is applied with $V=\alpha \ln\left( d\left( \frac{m-n}{N}\right)\right)$. To simplify notations, we can denote:
\begin{equation}\label{eq.CFT.V1V2}
\begin{split}
& V_1^{\alpha,m,n}
:= V_1\left(\alpha \ln\left( d\left( \frac{m-n}{N}\right)\right)\right)	
= \frac{1}{2} \left(\mathrm{sech}^{-1}\left( d\left(\frac{m-n}{N}\right)^{-\alpha}\right)+\mathrm{sech}^{-1}\left( d\left(\frac{m-n}{N}\right)^{\alpha}\right)\right)
\end{split}
\end{equation}
and
\begin{equation}
\begin{split}
&V_2^{\alpha,m,n} 
:= V_2\left(\alpha \ln\left( d\left( \frac{m-n}{N}\right)\right)\right)	
= \frac{1}{2} \left(\mathrm{sech}^{-1}\left( d\left(\frac{m-n}{N}\right)^{-\alpha}\right)-\mathrm{sech}^{-1}\left( d\left(\frac{m-n}{N}\right)^{\alpha}\right)\right)
\end{split}
\end{equation}
where the functions $V_1(x)$ and $V_2(x)$ are defined in Eq.~\eqref{eq.V1V2.C=1/2}.
Then Eq.~\eqref{eq.ssV2RBM.oneterm} becomes:
\begin{equation}
d\left( \frac{m-n}{N} \right)^{\alpha s_m s_n}	
=
\frac{1}{2}
\sum_{h_{m,n}=\pm 1} \exp\left( s_m V_1^{\alpha,m,n} h_{m,n}+s_n V_2^{\alpha,m,n} h_{m,n}\right)
\end{equation}
Applying this identity to each terms in Eq.~\eqref{eq.CFT.spinchain}, the wave function is transformed to:
\begin{equation}\label{eq.CFT.spinchain.toRBM.1}
\begin{split}
&\psi_{\mathrm{CFT}}(s_1,s_2,\ldots,s_N) 	\\
\propto& \delta\left( \sum_i s_i\right) \exp\left(\frac{\pi\ii}{2} \sum_{i:\mathrm{odd}} s_i-1 \right) \prod_{m<n}
\frac{1}{2}\sum_{h_{m,n}=\pm 1} \exp\left( s_m V_1^{\alpha,m,n} h_{m,n}+s_n V_2^{\alpha,m,n} h_{m,n}\right) 	\\
\propto& \delta\left( \sum_i s_i\right) \exp\left(\frac{\pi\ii}{2} \sum_{i:\mathrm{odd}} s_i \right) 
\sum_{\lbrace h_{m,n} \rbrace =\pm 1} \exp\left( \sum_{m<n} \left( s_m V_1^{\alpha,m,n} h_{m,n}+s_n V_2^{\alpha,m,n} h_{m,n} \right)\right)	\\
\end{split}
\end{equation}
where in the last equality, we rewrite the product of summation as a big summation, and the overall constant $\frac{1}{2}$ and the phase $\exp\left(-\frac{\pi\ii}{2} \right)$ are dropped. The last two terms in the last equality of Eq.~\eqref{eq.CFT.spinchain.toRBM.1} are already in the form of an RBM state. The rest of the task is to transform the $\delta$ function to an RBM state. We can achieve it by a Fourier transformation:
\begin{equation}\label{eq.CFT.fourier}
\begin{split}
\delta\left( \sum_i s_i\right) 
=& \frac{1}{N} \sum_{h=0,1,\ldots,N-1} \exp\left( \frac{\pi\ii}{N} h \sum_{i} s_i \right), \quad s_i \in \lbrace \pm 1\rbrace.
\end{split}
\end{equation}
Therefore, combining Eq.~\eqref{eq.CFT.fourier} and \eqref{eq.CFT.spinchain.toRBM.1}, $\psi_{\mathrm{CFT}}(s_1,s_2,\ldots,s_N)$ becomes:
\begin{equation}
\begin{split}
&\psi_{\mathrm{CFT}}(s_1,s_2,\ldots,s_N) 	\\
\propto& \frac{1}{N} \sum_{h=0,1,\ldots,N-1} \exp\left( \frac{\pi\ii}{N} h \sum_{i} s_i \right)
\exp\left(\frac{\pi\ii}{2} \sum_{i:\mathrm{odd}} s_i \right) 	
\sum_{\lbrace h_{m,n} \rbrace =\pm 1} \exp\left( \sum_{m<n} \left( s_m V_1^{\alpha,m,n} h_{m,n}+s_n V_2^{\alpha,m,n} h_{m,n} \right)\right) 	\\
=& \frac{1}{N} \sum_{h=0}^{N-1} \sum_{\lbrace h_{m,n} \rbrace =\pm 1}
\exp\left( \ii \left( \frac{\pi}{N} h \sum_{i} s_i + \frac{\pi}{2} \sum_{i:\mathrm{odd}} s_i - \ii \sum_{m<n} \left( s_m V_1^{\alpha,m,n} h_{m,n}+s_n V_2^{\alpha,m,n} h_{m,n} \right) \right) \right)
\end{split}
\end{equation}
This wave function is in the form of the RBM state. 

\setcounter{equation}{0}
\renewcommand{\theequation}{D\thechapter.\arabic{equation}}
\renewcommand{\thetheorem}{D\thechapter.\arabic{theorem}}

\section{D. MPS by 2-Layer RBM}
\label{app.proofForMPS2RBM}

In this appendix, we show that any MPS with finite bond dimension  can be approximated by a 2-layer RBM within an error which scales linearly with $N$.  Due to the length of the section, we briefly summarize the procedures:
\begin{enumerate}
\item Review the theorem in Ref.~\cite{le2008representational} where the authors proved that any (real) probability distribution can be approximated by a 1-layer RBM, i.e. Theorem \ref{theorem.LeRouxBengio}.
\item Establish Theorem \ref{Theorem1} that any complex function can be approximated by a 2-layer RBM.
\item Apply Theorem \ref{Theorem1} to the local complex tensor $A$ and thus obtain the approximated local 2-layer RBM, which is shown in Corollary \ref{corollary.MPS_Tensor_2LayerRBM}.
\item Combine each of the approximated local 2-layer RBM, which turns out to be a 2-layer RBM, i.e. Theorem \ref{theorem.Final}.
\end{enumerate}

\subsection{Conventions}

Suppose we have an MPS with periodic boundary condition whose local tensor is $A$:
\begin{equation}\label{eq.MPS}
\sum_{S_1,S_2,\ldots,S_N} \sum_{I_1,I_2,\ldots,I_N} A^{S_1}_{I_1,I_2} A^{S_2}_{I_2,I_3} \ldots A^{S_N}_{I_N,I_1} \ket{S_1,S_2,\ldots,S_N}
\end{equation}
The tensor $A$ has one physical index of dimension $d$ and two virtual indices of dimension $D$. The diagram representation for the tensor $A$ is:
\begin{equation}
\begin{gathered}
\includegraphics[width=2.5cm]{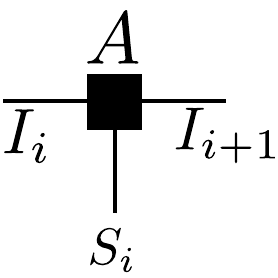}
\end{gathered}
\end{equation}
where each line represents an index of the tensor. For convenience, we assume:
\begin{equation}
d = 2^n,	\quad
D = 2^m.
\end{equation}
so that the indices of the tensor $A^{S}_{IJ}$ can be represented by the binary number indices:
\begin{equation}
\begin{split}
S &= (s_1s_2 \ldots s_n)_2, \quad \forall \; S = 0,1,2,\ldots,d-1.	\\
I &= (i_1i_2 \ldots i_m)_2, \quad \forall \; I = 0,1,2,\ldots,D-1.
\end{split}
\end{equation}
where each binary number is valued in $s_k\in\{0,1\}$ and $i_k \in \{0,1\}$. It is convenient to denote a 1-layer RBM abstractly without explicitly labeling the weights and biases，
\begin{equation}
\mathrm{RBM}_{n,m}(s_1,s_2,\ldots,s_n) 
= \sum_{\{h_j\} \in \{0,1\}^m } \exp\left( \sum_{ij} W_{ij}s_ih_j + \sum_{i} a_i s_i + \sum_{j} b_j h_j \right), \quad
\forall\; s_1,s_2,\ldots,s_n \in \{0,1\}^n,
\end{equation}
For the purpose of proving the existence of such an RBM, finding the specific values of weights and biases is not important, and the schematic notation of RBM $\mathrm{RBM}_{n,m}(s_1,s_2,\ldots,s_n) $ simplifies the discussion below. Similarly, we shall use the notation $\mathrm{RBM}_{n,p,q}\left(s_1,s_2,\ldots,s_n\right)$ to denote a 2-layer RBM whose number of visible spins is $n$, the number of the first layer hidden spins is $p$ and the number of the second layer hidden spins is $q$.

The 1-layer RBM has the following properties:
\begin{property}\label{property.RBM_notation}
\begin{enumerate}
\item Suppose $f_i(s_i)$ are linear functions of $s_i\; (i=1,2,\ldots,n)$ ($f_i(s_i)= p_i s_i + q_i,\; (i=1,2,\ldots,n)$), then $\mathrm{RBM}_{n,m}(f_1(s_1),f_2(s_2),\ldots,f_n(s_n))$ is also a 1-layer RBM.
\item The multiplication of two 1-layer RBM's of the same set of visible spins, with the number of hidden spins $m_1$ and $m_2$, is also a 1-layer RBM, however with of $m_1+m_2$ number of hidden spins
\begin{equation}
\mathrm{RBM}^{(1)}_{n,m_1}(s_1,s_2,\ldots,s_n)\mathrm{RBM}^{(2)}_{n,m_2}(s_1,s_2,\ldots,s_n)=\mathrm{RBM}_{n,m_1+m_2}(s_1,s_2,\ldots,s_n).
\end{equation}
$\mathrm{RBM}^{(1)}_{n,m_1}(s_1,s_2,\ldots,s_n)$ and $\mathrm{RBM}^{(2)}_{n,m_2}(s_1,s_2,\ldots,s_n)$ can have different weights and biases. 
\item The multiplication of two 1-layer RBM's $\mathrm{RBM}^{(1)}_{n_1,m_1}(s^{(1)}_1,s^{(1)}_2,\ldots,s^{(1)}_{n_1})$ and $\mathrm{RBM}^{(2)}_{n_2,m_2}(s^{(2)}_{1},\ldots,s^{(2)}_{n_2})$ is a  1-layer RBM of $n_1+n_2$  visible spins and $m_1+m_2$  hidden spins, i.e., 
\begin{equation}
\mathrm{RBM}^{(1)}_{n_1,m_1}(s^{(1)}_1,s^{(1)}_2,\ldots,s^{(1)}_{n_1})\mathrm{RBM}^{(2)}_{n_2,m_2}(s^{(2)}_{1},\ldots,s^{(2)}_{n_2})=\mathrm{RBM}_{n_1+n_2,m_1+m_2}(s^{(1)}_1,\ldots, s^{(1)}_{n_1},s^{(2)}_1,\ldots,s^{(2)}_{n_2}).
\end{equation}
$\mathrm{RBM}^{(1)}_{n,m_1}(s_1,s_2,\ldots,s_n)$ and $\mathrm{RBM}^{(2)}_{n,m_2}(s_1,s_2,\ldots,s_n)$ generally have different weights and biases. 
\end{enumerate}
\end{property}

\begin{proof}

1. Suppose the 1-layer RBM is:
\begin{equation}
\mathrm{RBM}_{n,m}(s_1,s_2,\ldots,s_n) = \sum_{h \in \{0,1\}^{m}} \exp\left( \sum_{ij} W_{ij}s_ih_j + \sum_{i} a_i s_i + \sum_{j} b_j h_j \right)
\end{equation}
and the linear functions are
\begin{equation}
f_i (x) = p_i x + q_i, \quad \forall\; i=1,2,\ldots,n.
\end{equation}
Then,
\begin{equation}
\begin{split}
& \mathrm{RBM}_{n,m}(f_1(s_1),f_2(s_2),\ldots,f_n(s_n)) \\
=& \sum_{h} \exp\left( \sum_{ij} W_{ij}(p_is_i+q_i) h_j + \sum_{i} a_i (p_is_i+q_i) + \sum_{j} b_j h_j \right)	\\
=& \sum_{h} \exp\left( \sum_{ij} \tilde{W}_{ij}s_ih_j + \sum_{i} \tilde{a}_i s_i + \sum_{j} \tilde{b}_j h_j + \sum_{i} a_i q_i \right)	\\
\end{split}
\end{equation}
where we have redefined:
\begin{equation}
\begin{split}
\tilde{W}_{ij} = W_{ij} p_{i},	\quad
\tilde{a}_{i} = a_{i} q_{i},	\quad
\tilde{b}_{j} = b_{j} + \sum_{i} W_{ij} q_{i},
\end{split}
\end{equation}
which is also a 1-layer  RBM. 

2. Suppose general expressions
\begin{equation}
\begin{split}
&\mathrm{RBM}^{(1)}_{n,m_1}(s_1,s_2,\ldots,s_n) = \sum_{h^{(1)}\in\{0,1\}^{m_1}} \exp\left( \sum_{ij} W^{(1)}_{ij}s_ih^{(1)}_j + \sum_{i} a^{(1)}_i s_i + \sum_{j} b^{(1)}_j h^{(1)}_j \right)	\\
&\mathrm{RBM}^{(2)}_{n,m_2}(s_1,s_2,\ldots,s_n)	= \sum_{h^{(2)}\in\{0,1\}^{m_2}} \exp\left( \sum_{ij} W^{(2)}_{ij}s_ih^{(2)}_j + \sum_{i} a^{(2)}_i s_i + \sum_{j} b^{(2)}_j h^{(2)}_j \right) \\	
\end{split}
\end{equation}
Hence, their product is
\begin{equation}
\begin{split}
&\mathrm{RBM}^{(1)}_{n,m_1}(s_1,s_2,\ldots,s_n)\mathrm{RBM}^{(2)}_{n,m_2}(s_1,s_2,\ldots,s_n)	\\
=& \sum_{h^{(1)}\in\{0,1\}^{m_1}} 
\exp\left( \sum_{ij} W^{(1)}_{ij}s_ih^{(1)}_j + \sum_{i} a^{(1)}_i s_i + \sum_{j} b^{(1)}_j h^{(1)}_j \right) 
\sum_{h^{(2)}\in\{0,1\}^{m_2}} 
\exp\left( \sum_{ij} W^{(2)}_{ij}s_ih^{(2)}_j + \sum_{i} a^{(2)}_i s_i + \sum_{j} b^{(2)}_j h^{(2)}_j \right)	\\
=& \sum_{h^{(1)}\in\{0,1\}^{m_1}} \sum_{h^{(2)}\in\{0,1\}^{m_2}} 
\exp\left( \sum_{ij} W^{(1)}_{ij}s_ih^{(1)}_j + \sum_{i} a^{(1)}_i s_i + \sum_{j} b^{(1)}_j h^{(1)}_j + \sum_{ij} W^{(2)}_{ij}s_ih^{(2)}_j + \sum_{i} a^{(2)}_i s_i + \sum_{j} b^{(2)}_j h^{(2)}_j \right)
\end{split}
\end{equation}
which is a 1-layer RBM with $m_1+m_2$ hidden spins.
	
3. The proof is essentially the same as the proof of the second property. Suppose
\begin{equation}
\begin{split}
\mathrm{RBM}^{(1)}_{n_1,m_1}(s^{(1)}_1,s^{(1)}_2,\ldots,s^{(1)}_{n_1}) &= \sum_{h^{(1)}\in\{0,1\}^{m_1}} \exp\left( \sum_{ij} W^{(1)}_{ij}s^{(1)}_ih^{(1)}_j + \sum_{i} a^{(1)}_i s^{(1)}_i + \sum_{j} b^{(1)}_j h^{(1)}_j \right)	\\
\mathrm{RBM}^{(2)}_{n_2,m_2}(s^{(2)}_1,s^{(2)}_2,\ldots,s^{(2)}_{n_2}) &= \sum_{h^{(2)}\in\{0,1\}^{m_2}} \exp\left( \sum_{ij} W^{(2)}_{ij}s^{(2)}_ih^{(2)}_j + \sum_{i} a^{(2)}_i s^{(2)}_i + \sum_{j} b^{(2)}_j h^{(2)}_j \right) \\
\end{split}
\end{equation}
Hence, direct multiplication gives
\begin{equation}
\begin{split}
&\mathrm{RBM}^{(1)}_{n_1,m_1}(s^{(1)}_1,s^{(1)}_2,\ldots,s^{(1)}_n)\mathrm{RBM}^{(2)}_{n_2,m_2}(s^{(2)}_1,s^{(2)}_2,\ldots,s^{(2)}_n)	\\
=& \sum_{h^{(1)}\in\{0,1\}^{m_1}} \exp\left( \sum_{ij} W^{(1)}_{ij}s^{(1)}_ih^{(1)}_j + \sum_{i} a^{(1)}_i s^{(1)}_i + \sum_{j} b^{(1)}_j h^{(1)}_j \right)	\\
&\sum_{h^{(2)}\in\{0,1\}^{m_2}} 
\exp\left( \sum_{ij} W^{(2)}_{ij}s^{(2)}_ih^{(2)}_j + \sum_{i} a^{(2)}_i s^{(2)}_i + \sum_{j} b^{(2)}_j h^{(2)}_j \right)	\\
=& \sum_{h^{(1)}\in\{0,1\}^{m_1}} \sum_{h^{(2)}\in\{0,1\}^{m_2}} 
\exp\left( \sum_{ij} W^{(1)}_{ij}s^{(1)}_ih^{(1)}_j + \sum_{i} a^{(1)}_i s^{(1)}_i + \sum_{j} b^{(1)}_j h^{(1)}_j + \sum_{ij} W^{(2)}_{ij}s^{(2)}_ih^{(2)}_j + \sum_{i} a^{(2)}_i s^{(2)}_i + \sum_{j} b^{(2)}_j h^{(2)}_j \right)
\end{split}
\end{equation}
which is a 1-layer RBM with $m_1+m_2$  hidden spins and $n_1+n_2$ visible spins.
\end{proof}

\subsection{Review of Useful Theorems and Identities}
\label{sec.useful}

Before discussing the representing power of the RBM, we first introduce the definition of the Kullback-Leibler (KL) divergence which measures how close two probability distributions are. The reader can find the general definition of the KL divergence in Ref.~\cite{wiki:KLdiv}, which can be defined for probability distributions over both continuous and discrete variables.

\begin{definition}
Suppose $f_I(\{s_i\})$, $I=1,2$ are two probability distributions over a set of $n$ $\mathbb{Z}_2$ valued variables $s_i\in \{0, 1\}$, $i=1, ..., n$, satisfying
\begin{eqnarray}\label{Eq.4}
\sum_{\{s_i\}} f_I(\{s_i\})=1, ~~~~~0\leq f_I(\{s_i\})\leq 1
\end{eqnarray}
The Kullback-Leibler (KL) divergence between $f_1(\{s_i\})$ and $f_2(\{s_i\})$ is defined as 
\begin{eqnarray}
KL(f_1|f_2)= \sum_{\{s_i\}} f_1(\{s_i\}) \log \frac{f_1(\{s_i\})}{f_2(\{s_i\})}
\end{eqnarray}
\end{definition}
Two important properties of the KL divergence are that 
(1) it is always non-negative due to Jensen's inequality \cite{wiki:KLdiv};
(2) it satisfies Pinsker's inequality \cite{wiki:Pinsker}. Concretely,  the square root of $KL$ gives the upper bound of the \emph{statistical difference} of $f_1$ and $f_2$, 
\begin{equation}
    \max_{\{s_i\}}|f_1(\{s_i\}- f_2(\{s_i\})| \leq \sqrt{\frac{1}{2}KL(f_1|f_2)}
\end{equation}
We proceed to present a theorem on the representing power of the 1-layer RBM, which was originally presented and proved in Theorem 2.4 of Ref.~\cite{le2008representational}. We rephrase their theorem as follows ( A typo of their manuscript is corrected and an equation is simplified further. )
\begin{theorem}\label{theorem.LeRouxBengio}
Suppose that $f$ is a probability distribution over the variables $s_i \in \lbrace 0,1 \rbrace, \; i=1,2,\ldots,n$. Let $k$ be the number of variables whose probability distribution is nonzero, i.e., $k$ is the volume of the set 
\begin{equation}
\{ \{s_i\}|f\left(\lbrace s_i \rbrace \right) > 0\}
\end{equation}
and $\xi$ is the minimal nonzero probability
\begin{equation}
\xi = \min_{ \{s_i\}, f\left(\lbrace s_i \rbrace \right) > 0} f\left(\lbrace s_i \rbrace \right).
\end{equation}
Then for any $\epsilon>0$, there exists a 1-layer RBM,  $\mathrm{RBM}_{n,k}^{\lambda_1}\left( \lbrace s_i \rbrace \right)$,  labeled by $\lambda_1 > 0$ with $k$ hidden spins such that the KL divergence between $f$ and $\mathrm{RBM}_{n,k}^{\lambda_1}$ is controlled: 
\begin{equation}
\mathrm{KL}(f|\mathrm{RBM}_{n,k}^{\lambda_1}) = \ln\left( 1+ \frac{\left( 2^n - k \right) \xi}{1 + \exp\left(\lambda_1\right)} \right) < \epsilon, ~~~~~ \text{for finite }n.
\end{equation}
\end{theorem}
We refer the reader to Ref.~\cite{le2008representational} for the definition of $\lambda_1$. The specific construction of $\lambda_1$ will not be important for our purposes. 
Using Pinsker's inequality, Theorem \ref{theorem.LeRouxBengio} implies that:
\begin{equation}\label{eq.PinskerCorollary}
\max_{\lbrace s_i \rbrace} | f\left(\lbrace s_i \rbrace \right) -\mathrm{RBM}_{n,k}^{\lambda_1}\left(\lbrace s_i \rbrace \right)| 
\le \sqrt{\frac{1}{2} \mathrm{KL}(f|\mathrm{RBM}_{n,k}^{\lambda_1}) } \le \sqrt{\frac{1}{2} \epsilon}
\end{equation}
Hence,
\begin{equation}\label{eq.PinskerAndLimit}
\lim\limits_{\lambda_1 \rightarrow \infty} \mathrm{RBM}_{n,k}^{\lambda_1}\left(\lbrace s_i \rbrace \right) = f\left(\lbrace s_i \rbrace \right),  \quad \forall\; s_i=0,1, \quad \forall\; i=1,2,\ldots,n.
\end{equation}

Another useful identity was introduced in Ref.~\cite{carleo2018constructing}.

\begin{theorem}\label{theorem.CarleoNomuraImada}
(Carleo, Nomura and Imada (2018), Eq. (22)) 
$\exp\left( \sigma_1 \sigma_2 \ldots \sigma_N V\right)$, where $\sigma_i \in \{\pm 1\}$, can be expressed in terms of a 1-layer RBM with $N$ visible spins and 2 hidden spins. 
\begin{equation}\label{eq.CarleoNomuraImada}
\begin{split}
&\exp\left( \sigma_1 \sigma_2 \ldots \sigma_N V\right) \\
=& C \cos^2 \left( b + \frac{\pi}{4} \sum_{i=1}^{N}\sigma_i \right)	\\
=& \frac{C}{4} \sum_{h_1,h_2=\pm 1} \exp\left( \ii \left( b + \frac{\pi}{4} \sum_{i=1}^{N}\sigma_i\right) \left(h_1+h_2\right) \right)	\\
=& \frac{C}{4} \sum_{\tilde{h}_1,\tilde{h}_2=0,1} \exp\left( \ii \left( b + \frac{\pi}{4} \sum_{i=1}^{N}\sigma_i\right) \left(1-2\tilde{h}_1+1-2\tilde{h}_2\right) \right)\\
=& \mathrm{RBM}_{N,2}(\sigma_1, ..., \sigma_N), 
\quad \forall\; \sigma_i \in \{\pm 1\}, i=1,2,\ldots,N.
\end{split}
\end{equation}
where 
\begin{equation}
\begin{split}
b &= \arctan \left( e^{-V} \right) - \frac{\pi}{4}\mathrm{mod}(N,4)	\\
C &= \frac{1}{\cos\left(\arctan\left(e^{-V}\right)\right) \sin\left(\arctan\left(e^{-V}\right)\right)}	\\
\end{split}
\end{equation}
\end{theorem}

In the next subsection, we will only use the schematic equivalence $\exp\left( \sigma_1 \sigma_2 \ldots \sigma_N V\right) = \mathrm{RBM}_{N,2}(\sigma_1, ..., \sigma_N)$, i.e., the last equality of Eq.~\eqref{eq.CarleoNomuraImada}. For definiteness, we will assume that all the hidden spins in this appendix are $\{0,1\}$ valued. Note that this schematic relation might use more hidden spins than necessary. For instance, when $N=2$, as demonstrated in Lemma \ref{lemma.ssV2RBM}, it suffices to introduce only one hidden spin, $h$. 

It will be convenient to convert Theorem \ref{theorem.CarleoNomuraImada} from $\lbrace \pm1 \rbrace$ convention to $\lbrace 0,1 \rbrace$ convention to label spins. 
In particular, the cases of $N=2,3,4$ are extensively used in our following calculations.

\begin{lemma}\label{lemma.2body}
$\exp(Vs_1s_2)$, $s_i=0,1$, can be expressed as a 1-layer RBM with $2$ visible spins and 2 hidden spin, i.e, 
\begin{eqnarray}\label{Eq.2}
\exp(Vs_1s_2)= \mathrm{RBM}_{2,1}(s_1, s_2), \forall\; s_1,s_2 = 0,1.
\end{eqnarray}
\end{lemma}

\begin{proof}
Denoting $\sigma_i=2s_i-1\in\{\pm 1\}$, then 
$\exp(Vs_1s_2)=\exp(V \sigma_1\sigma_2/4)\exp(V (\sigma_1+ \sigma_2)/4) \exp(V/4)$. Using Eq.~\eqref{eq.ssV2RBM}, we have $\exp(V \sigma_1\sigma_2/4)= \mathrm{RBM}_{2,1}(\sigma_1,\sigma_2)$. Further multiplying $\exp(V (\sigma_1+ \sigma_2)/4)$ amounts to adding bias terms in the RBM. Hence it immediately follows that $\exp(Vs_1s_2)$ can also be expressed as an RBM with 2 visible spins and 1 hidden spin. In conclusion, Eq.~\eqref{Eq.2} holds. 
\end{proof}

\begin{lemma}\label{lemma.3bodyTo5hidden}
$\exp\left( s_1 s_2 s_3 V\right) $ with $ s_i\in \{0,1\}$ can be written as a 1-layer RBM with 3 visible spins and 5 hidden spins, i.e., 
\begin{eqnarray}
\exp\left( s_1 s_2 s_3 V\right) = \mathrm{RBM}_{3,5}\left( s_1,s_2,s_3 \right)
\end{eqnarray}
\end{lemma}

\begin{proof}
For simplicity, we will use Theorem \ref{theorem.CarleoNomuraImada} and Corollary \eqref{lemma.2body} to prove this lemma. Since the visible spins are $\{0,1\}$ valued, to use Theorem \ref{theorem.CarleoNomuraImada}, we introduce the $\{\pm 1\}$ valued variables $\sigma_i = 2 s_i -1 \in \{\pm 1\}, \quad i=1,2,3.$ Then $\exp\left( s_1 s_2 s_3 V\right) $ can be expressed in terms of $\sigma_i$ as 
\begin{equation}
\begin{split}
&\exp\left( s_1 s_2 s_3 V\right)
=\exp\left( \frac{V}{8} \sigma_1 \sigma_2 \sigma_3 \right) \exp\left( \frac{V}{8}\left(\sigma_1 \sigma_2 + \sigma_1 \sigma_3 + \sigma_2 \sigma_3 \right) \right) \exp\left( \frac{V}{8} \left( \sigma_1 + \sigma_2 + \sigma_3 + 1 \right)\right)	\\
\end{split}
\end{equation}
According to Theorem \ref{theorem.CarleoNomuraImada} and Lemma \ref{lemma.2body}, we can introduce an RBM for each of the terms in the above equation:
\begin{equation}
\begin{split}
&\exp\left( \frac{V}{8} \sigma_1 \sigma_2 \sigma_3 \right)= \mathrm{RBM}_{3,2}(\sigma_1,\sigma_2,\sigma_3), \\
&\exp\left( \frac{V}{8} \left(\sigma_1 \sigma_2 + \sigma_1 \sigma_3 + \sigma_2 \sigma_3 \right)+ \frac{V}{8} \left( \sigma_1 + \sigma_2 + \sigma_3 + 1 \right)\right)	\\
&=\frac{1}{2} \sum_{h_1, h_2, h_3=\pm 1} \exp(V_1 \sigma_1 h_1+ V_2 \sigma_2 h_1+ V_1 \sigma_2 h_2+ V_2 \sigma_3 h_2+ V_1 \sigma_3 h_3 + V_2 \sigma_1 h_3+ \frac{V}{8} \left( \sigma_1 + \sigma_2 + \sigma_3 + 1 \right))\\
&= \mathrm{RBM}_{3,3}(\sigma_1,\sigma_2, \sigma_3)
\end{split}
\end{equation}
we find
\begin{equation}
\exp\left( s_1 s_2 s_3 V\right)
= \mathrm{RBM}_{3,2}(\sigma_1,\sigma_2,\sigma_3) \mathrm{RBM}_{3,3}(\sigma_1,\sigma_2, \sigma_3)
\end{equation}
According to Property \ref{property.RBM_notation}, this becomes an RBM with 5 hidden spins.
\begin{equation}
\exp\left( s_1 s_2 s_3 V\right) = \mathrm{RBM}_{3,5}\left( \sigma_1,\sigma_2,\sigma_3 \right) = \mathrm{RBM}_{3,5}\left( s_1,s_2,s_3 \right)
\end{equation}
Thus $\exp\left( s_1 s_2 s_3 V\right)$ can be expressed in the form of a 1-layer RBM with 3 visible spins  and 5 hidden spins. 
\end{proof}

\begin{lemma}\label{lemma.4bodyTo16hidden}
$\exp\left( s_1 s_2 s_3s_4 V\right) $ with $ s_i\in \{0,1\}$ can be written as a 1-layer RBM with 4 visible spins and 16 hidden spins, i.e., 
\begin{eqnarray}
\exp\left( s_1 s_2 s_3 s_4V\right) = \mathrm{RBM}_{4,16}\left( s_1,s_2,s_3,s_4 \right)
\end{eqnarray}
\end{lemma}

\begin{proof}
Since the proof is similar to proof of Theorem \ref{lemma.3bodyTo5hidden}, we will be brief here. We introduce $\{\pm 1\}$ valued variables $\sigma_i = 2 s_i -1 \in \{\pm 1\}, \quad i=1,2,3, 4.$ Then $\exp\left( s_1 s_2 s_3 s_4 V\right) $ can be expressed in terms of $\sigma_i$ as 	
\begin{equation}
\begin{split}
&\exp\left( s_1 s_2 s_3 V\right)
=\exp\left( \frac{V}{16} \sigma_1 \sigma_2 \sigma_3 \sigma_4\right) \exp\left( \frac{V}{16}\sum_{1=i<j<k=4}\sigma_i \sigma_j \sigma_k \right) \exp\left( \frac{V}{16}\sum_{1=i<j=4}\sigma_i \sigma_j \right)\exp\left( 1+ \frac{V}{16}\sum_{1=i=4}\sigma_i \right)	\\
\end{split}
\end{equation}
According to Theorem \ref{theorem.CarleoNomuraImada} and Lemma \ref{lemma.2body}, we can introduce an RBM for each of the terms in the above equation:
\begin{equation}
\begin{split}
&\exp\left( \frac{V}{8} \sigma_1 \sigma_2 \sigma_3 \sigma_4\right)= \mathrm{RBM}_{4,2}(\sigma_1,\sigma_2,\sigma_3, \sigma_4), \\
&\exp\left( \frac{V}{16}\sum_{1=i<j<k=4}\sigma_i \sigma_j \sigma_k \right)= \mathrm{RBM}_{4,8}(\sigma_1,\sigma_2, \sigma_3, \sigma_4)\\
&\exp\left( \frac{V}{16}\sum_{1=i<j=4}\sigma_i \sigma_j \right)\exp\left( 1+ \frac{V}{16}\sum_{1=i=4}\sigma_i \right)= \mathrm{RBM}_{4,6}(\sigma_1,\sigma_2, \sigma_3, \sigma_4)\\
\end{split}
\end{equation}
we find
\begin{equation}
\begin{split}
\exp\left( s_1 s_2 s_3s_4 V\right)
&= \mathrm{RBM}_{4,2}(\sigma_1,\sigma_2,\sigma_3,\sigma_4) \mathrm{RBM}_{4,8}(\sigma_1,\sigma_2, \sigma_3, \sigma_4) \mathrm{RBM}_{4,6}(\sigma_1,\sigma_2, \sigma_3, \sigma_4)\\
&= \mathrm{RBM}_{4,16}(\sigma_1,\sigma_2, \sigma_3, \sigma_4)= \mathrm{RBM}_{4,16}(s_1,s_2, s_3, s_4)
\end{split}
\end{equation}
Thus $\exp\left( s_1 s_2 s_3 s_4V\right)$ can be expressed in the form of a 1-layer RBM with 4 visible spins and 16 hidden spins. 
\end{proof}

\subsection{MPS approximated by 2-Layer RBM}
\label{sec.tensor2RBM}

Theorem \ref{theorem.LeRouxBengio} only discuss the approximation of a probability distribution by an 1-layer RBM. We will first discuss, in theorem \ref{theorem.realTensorApproximation}, the approximation of a real tensor (which need not be a probability distribution) by a 2-layer RBM. We will discuss the approximation of a complex tensor in theorem \ref{Theorem1}. 

\begin{theorem}\label{theorem.realTensorApproximation}
Suppose a real tensor $A_{s_1s_2\ldots s_n}$ has binary indices $s_i\in \{0,1\}, \; (i=1,2,\ldots,n)$. Denote $k^+$ is the number of elements of $A_{s_1s_2...s_n}$ satisfying $A_{s_1s_2...s_n}>0$, $k^-$ is the number of elements of $A_{s_1s_2...s_n}$ satisfying $A_{s_1s_2...s_n}<0$, and $k=\max \{k^+, k^-\}$.  Then for any $\epsilon>0$, there exists a 2-layer RBM  $\mathrm{RBM}_{n,6nk+n+k, k+1} (\{s_i\})$ 
such that:
\begin{equation}
\max_{\lbrace s_i \rbrace} | A_{s_1s_2\ldots s_n} -\mathrm{RBM}_{n,6nk+n+k, k+1} (\{s_i\})| \le \epsilon.
\end{equation} 
where there are $n$ visible spins, $6nk+n+k$ hidden spins in the first layer, and $k+1$ hidden spins in the second layer. 
\end{theorem}

\begin{proof}
Suppose the real tensor is decomposed as follow:
\begin{equation}
A_{s_1s_2\ldots s_n} = A^{+}_{s_1s_2\ldots s_n} - A^{-}_{s_1s_2\ldots s_n},
\end{equation}
where
\begin{equation}
\begin{split}
A^{+}_{s_1s_2\ldots s_n} &= \begin{cases}
A_{s_1s_2\ldots s_n} & \text{if } s_1,s_2\dots,s_n \text{ satisfy } A_{s_1s_2\ldots s_n} > 0,	\\
0	& \text{otherwise},
\end{cases} 	\\
A^{-}_{s_1s_2\ldots s_n} &= \begin{cases}
-A_{s_1s_2\ldots s_n} & \text{if } s_1,s_2\dots,s_n \text{ satisfy } A_{s_1s_2\ldots s_n} < 0,	\\
0	& \text{otherwise}.
\end{cases} 
\end{split}
\end{equation}
Using Eq.~\eqref{eq.PinskerCorollary}, for any $\epsilon > 0$, there exist two 1-layer RBM's such that these two equations hold true:
\begin{equation}
\begin{split}
\max_{\lbrace s_i \rbrace} | A^{+}_{s_1s_2\ldots s_n} - {\mathcal{N}^{+}}\mathrm{RBM}^+_{n,k^+}\left( s_1,s_2,\ldots, s_n \right)| < \frac{1}{2} \epsilon,	\\
\max_{\lbrace s_i \rbrace} | A^{-}_{s_1s_2\ldots s_n} - {\mathcal{N}^{-}}\mathrm{RBM}^-_{n,k^-}\left( s_1,s_2,\ldots, s_n \right)| < \frac{1}{2} \epsilon.
\end{split}
\end{equation}
where
\begin{equation}
\mathcal{N}^{+} = {\sum_{\{s_i\}} A^{+}_{s_1s_2\ldots s_n}},	\quad
\mathcal{N}^{-} = {\sum_{\{s_i\}} A^{-}_{s_1s_2\ldots s_n}}.
\end{equation}
Hence, we have:
\begin{equation}
\begin{split}
&\max_{\lbrace s_i \rbrace} | A_{s_1s_2\ldots s_n} - {\mathcal{N}^{+}} \mathrm{RBM}^+_{n,k^+}\left( s_1,s_2,\ldots, s_n \right) + {\mathcal{N}^{-}} \mathrm{RBM}^-_{n,k^-}\left( s_1,s_2,\ldots, s_n \right)|	\\
=& \max_{\lbrace s_i \rbrace} | A^+_{s_1s_2\ldots s_n} - A^-_{s_1s_2\ldots s_n} - {\mathcal{N}^{+}} \mathrm{RBM}^+_{n,k^+}\left( s_1,s_2,\ldots, s_n \right) + {\mathcal{N}^{-}} \mathrm{RBM}^-_{n,k^-}\left( s_1,s_2,\ldots, s_n \right)|	\\
\le& \max_{\lbrace s_i \rbrace} \left( | A+_{s_1s_2\ldots s_n} - {\mathcal{N}^{+}} \mathrm{RBM}^+_{n,k^+}\left( s_1,s_2,\ldots, s_n \right)| + | A^-_{s_1s_2\ldots s_n} - {\mathcal{N}^{-}}\mathrm{RBM}^-_{n,k^-}\left( s_1,s_2,\ldots, s_n \right)| \right) 	\\
\le& \max_{\lbrace s_i \rbrace} | A+_{s_1s_2\ldots s_n} - {\mathcal{N}^{+}}\mathrm{RBM}^+_{n,k^+}\left( s_1,s_2,\ldots, s_n \right)| + \max_{\lbrace s_i \rbrace} | A^-_{s_1s_2\ldots s_n} - {\mathcal{N}^{-}}\mathrm{RBM}^-_{n,k^-}\left( s_1,s_2,\ldots, s_n \right)|	\\
\le& \epsilon
\end{split}
\end{equation}
We further show that  $ {\mathcal{N}^{+}}\mathrm{RBM}^+_{n,k^+} -  {\mathcal{N}^{-}}\mathrm{RBM}^-_{n,k^-}$ can be expressed as a 2-layer RBM. 
\begin{equation}\label{Eq.205}
\begin{split}
& {\mathcal{N}^{+}}\mathrm{RBM}^+_{n,k^+}\left( s_1,s_2,\ldots, s_n \right) -  {\mathcal{N}^{-}}\mathrm{RBM}^-_{n,k^-}\left( s_1,s_2,\ldots, s_n \right)	\\
=& {\mathcal{N}^{+}}\mathrm{RBM}^+_{n,k}\left( s_1,s_2,\ldots, s_n \right) - {\mathcal{N}^{-}}\mathrm{RBM}^-_{n,k}\left( s_1,s_2,\ldots, s_n \right)	\\
=& \sum_{\{h_1,h_2,\ldots h_k \}} \exp\left( W^+_{ij} s_i h_j + a^+_i s_i + b^+_j h_j + \ln\left( \mathcal{N}^+ \right)\right) + \exp\left( W^-_{ij} s_i h_j + a^-_i s_i + b^-_j h_j + \ln\left( \mathcal{N}^- \right) + \ii\pi\right)\\
=& \sum_{\{h_1,h_2,\ldots h_k \}} \sum_{h_0=0}^{1} e^{ h_0 \lbrack W^+_{ij} s_i h_j + a^+_i s_i + b^+_j h_j + \ln\left( \mathcal{N}^+ \right) - W^-_{ij} s_i h_j - a^-_i s_i - b^-_j h_j - \ln\left( \mathcal{N}^- \right) + \ii\pi\rbrack + W^-_{ij} s_i h_j + a^-_i s_i + b^-_j h_j + \ln\left( \mathcal{N}^- \right) + \ii\pi}	\\
=& \sum_{\{h_0, h_1,h_2,\ldots h_k \}} \prod_{ij} \mathrm{RBM}_{3,5}\left(h_0, s_i, h_j\right) \prod_{i} \mathrm{RBM}_{2,1}\left(h_0, s_i\right)
\prod_{j} \mathrm{RBM}_{2,1}\left(h_0, h_j\right) \prod_{ij} \mathrm{RBM}_{2,1}\left(s_i,h_j\right)	\\
=& \sum_{\{h_0, h_1,h_2,\ldots h_k \}} \mathrm{RBM}_{n+k+1, 6nk + n+ k}(s_1,s_2,\ldots,s_n,h_0,h_1,\ldots,h_k)
\end{split}	
\end{equation}
In the last second equality, we have used
\begin{eqnarray}
\begin{split}
\exp((W^+_{ij}-W^-_{ij})h_0s_ih_j) &= \mathrm{RBM}_{3,5}\left(h_0, s_i, h_j\right)\\
\exp((a_i^+-a_i^-)h_0s_i) &= \mathrm{RBM}_{2,1}\left(h_0, s_i\right)\\
\exp((b_j^+-b_j^-)h_0h_j - (\ln(\mathcal{N}^+)-\ln(\mathcal{N}^-)) h_0 + \ii\pi ) &=\mathrm{RBM}_{2,1}\left(h_0, h_j\right) \\
\exp(W^-_{ij} s_i h_j + a^-_i s_i + b^-_j h_j - \ln\left( \mathcal{N}^- \right) + \ii\pi)&= \mathrm{RBM}_{2,1}\left(s_i,h_j\right)	
\end{split}
\end{eqnarray}
\begin{figure}[H]
	\centering
	\includegraphics[width=0.6\columnwidth]{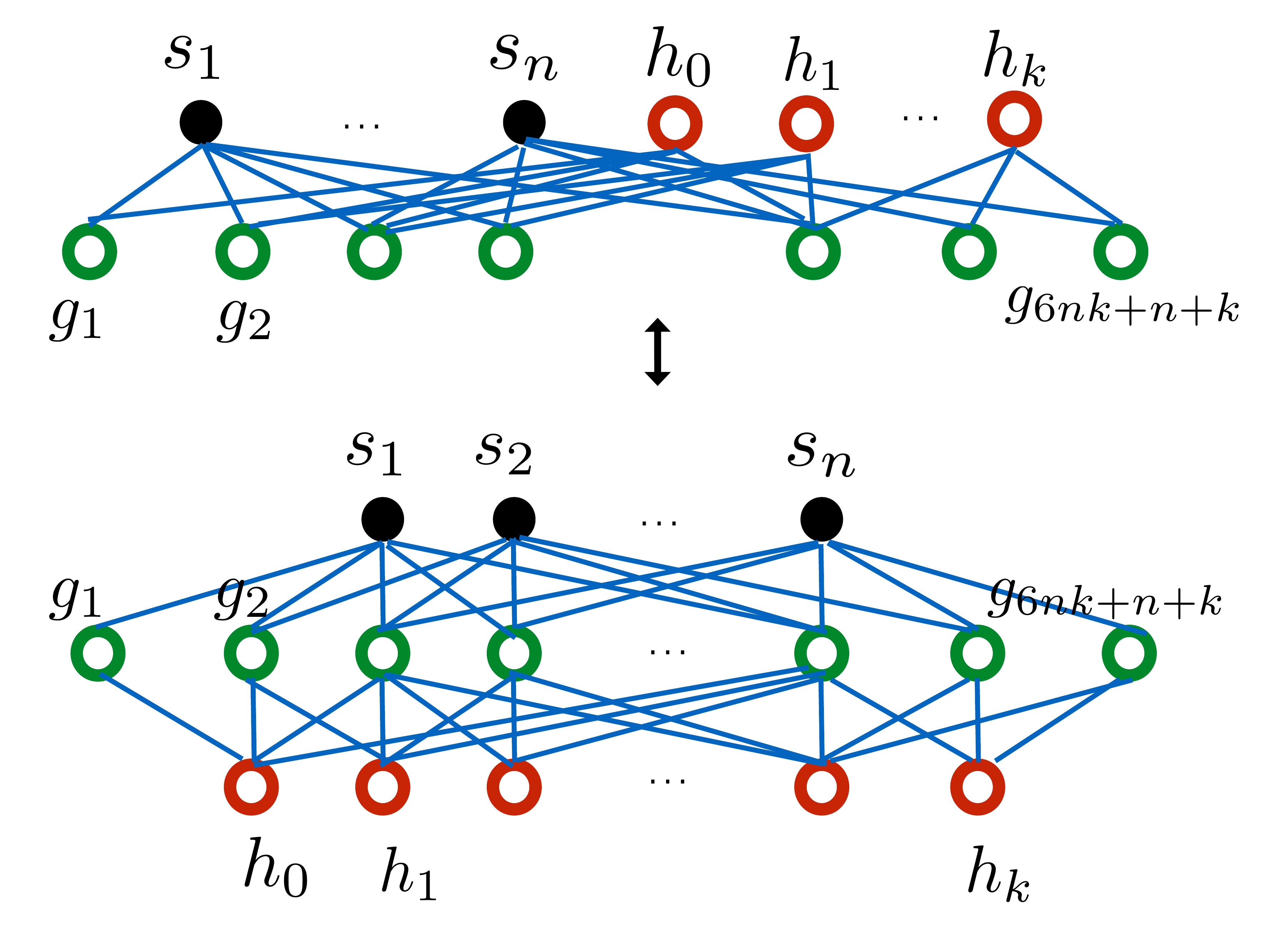}
	\caption{Top: A graphical representation of the 1-layer RBM $ \mathrm{RBM}_{n+k+1, 6nk + n+ k}(s_1,s_2,\ldots,s_n,h_0,h_1,\ldots,h_k)$. Black disks and red circles represent the $n+k+1$ visible spins of the 1-layer RBM $s_1, ..., s_n, h_0, h_1, ... h_k$, while the green circles represent the $6nk+n+k$ hidden spins. $g_1, ..., g_{6nk+n+k}$. Bottom:   A graphical representation of the 2-layer RBM $\mathrm{RBM}_{n,6nk+n+k, k+1} (\{s_i\})$.  This is obtained by summing over $h_0, h_1, ..., h_k$. Since $h$'s are summed over, they are regarded as the hidden spins in the second hidden layer. The graph in the bottom is obtained from the graph in the top by simply folding the red circles to the other side of the green circles. }
	\label{Fig.Real}
\end{figure}

To see that the last equality in Eq.~\eqref{Eq.205} is a 2-layer RBM, we notice that in the 1-layer RBM $\mathrm{RBM}_{n+k+1, 6nk + n+ k}(s_1,s_2,\ldots,s_n,h_0,h_1,\ldots,h_k)$,   $\{s_i\}$ and $\{h_j\}$ are all connected with $6nk + n+ k$ hidden spins, while $\{s_i\}$ and $\{h_j\}$ are not mutually connected. Using this structure, one can label $\{s_1, ..., s_n\}$ as the visible spin of the 2-layer RBM, the $6nk+n+k$ hidden spins in the above 1-layer RBM as the hidden spins in the first layer of the 2-layer RBM, and $k+1$ spins $\{h_0, h_1, ..., h_k\}$ as the hidden spins in the second layer of the 2-layer RBM. In summary, we denote the 2-layer RBM as 
\begin{eqnarray}
\mathrm{RBM}_{n,6nk+n+k, k+1} (\{s_i\})=  \sum_{\{h_0, h_1,h_2,\ldots h_k \}} \mathrm{RBM}_{n+k+1, 6nk + n+ k}(s_1,s_2,\ldots,s_n,h_0,h_1,\ldots,h_k)
\end{eqnarray}
See figure \ref{Fig.Real} for the graphical representation of the 2-layer RBM $\mathrm{RBM}_{n,6nk+n+k, k+1} (\{s_i\})$. 
This completes the proof of theorem \ref{theorem.realTensorApproximation}. 
\end{proof}

As we will see in theorem \ref{Theorem1}, we will keep $n$ to be a finite number, which does not go to infinity as the system size (i.e. the number of unit cells $N$) goes to infinity. This is because $n$ is the number of spins for each MPS matrix. $n$ should not be confused with $N$.

\begin{theorem}\label{Theorem1}
Suppose a complex tensor $A_{s_1s_2\ldots s_n}$ has binary indices $s_i\in \{0,1\}, \; (i=1,2,\ldots,n)$. Denote $k^{R\pm}$ is the number of elements of the real part of $A_{s_1s_2\ldots s_n}$ which are strictly positive and negative respectively, i.e., $\Re\left( A_{s_1s_2\ldots s_n}\right)>0$ or $\Re\left( A_{s_1s_2\ldots s_n}\right)<0$. Similarly, $k^{I\pm}$ is the number of elements of the imaginary part of $A_{s_1s_2\ldots s_n}$ which are strictly positive and negative respectively, i.e., $\Im\left( A_{s_1s_2\ldots s_n}\right)>0$ or $\Im\left( A_{s_1s_2\ldots s_n}\right)<0$. Then for any $\epsilon>0$, there exists a 2-layer RBM such that 
\begin{eqnarray}
\max_{\lbrace s_i \rbrace} | A_{s_1s_2\ldots s_n} - \mathrm{RBM}_{n, 36kn^2+6n^2+36k^2 n+ 54k n+ 8n+6k^2+8k+1, 6nk+n+2k+2}(\{s_i\}) | \le \epsilon.
\end{eqnarray}
where $k=\max(k^{R+},k^{R-},k^{I+},k^{I-})$. The 2-layer RBM contains $n$ visible spins, $36kn^2+6n^2+36k^2 n+ 54k n+ 8n+6k^2+8k+1$ hidden spins in the first hidden layer and $6nk+n+2k+2$ hidden spins in the second hidden layer.
\end{theorem}

\begin{proof}
Because $A_{s_1s_2\ldots s_n}$ is complex, to use Theorem \ref{theorem.realTensorApproximation}, we consider the real part and the imaginary part of $A_{s_1s_2\ldots s_n}$ respectively. Denote $\Re\left( A_{s_1s_2\ldots s_n} \right) $ and $\Im\left( A_{s_1s_2\ldots s_n} \right) $ as the real and imaginary part respectively, i.e., 
\begin{eqnarray}
A_{s_1s_2\ldots s_n}=\Re\left( A_{s_1s_2\ldots s_n} \right) + \ii \Im\left( A_{s_1s_2\ldots s_n} \right) 
\end{eqnarray}
Using Theorem \ref{theorem.realTensorApproximation}, for any given $\epsilon>0$, there exist 2-layer $\mathrm{RBM}^R_{n,6nk^R+n+k^R, k^R+1}$ and $\mathrm{RBM}^I_{n,6nk^I+n+k^I, k^I+1}$ such that
\begin{eqnarray}
\begin{split}
& \max_{\lbrace s_i \rbrace} |\Re\left( A_{s_1s_2\ldots s_n} \right)  -\mathrm{RBM}^R_{n,6nk^R+n+k^R, k^R+1}(\{s_i\}) |< \frac{\epsilon}{2}\\
& \max_{\lbrace s_i \rbrace} |\Im\left( A_{s_1s_2\ldots s_n} \right)  -\mathrm{RBM}^I_{n,6nk^I+n+k^I, k^I+1}(\{s_i\}) |< \frac{\epsilon}{2}
\end{split}
\end{eqnarray}
where $k^R = \max\left( k^{R+}, k^{R-}\right)$ and $k^I = \max\left( k^{I+}, k^{I-}\right)$.
Then
\begin{eqnarray}
\begin{split}
&\max_{\lbrace s_i \rbrace} | A_{s_1s_2\ldots s_n}  -\mathrm{RBM}^R_{n,6nk^R+n+k^R, k^R+1}(\{s_i\}) - i\mathrm{RBM}^I_{n,6nk^I+n+k^I, k^I+1}(\{s_i\})  |\\
&= \max_{\lbrace s_i \rbrace} |\Re\left( A_{s_1s_2\ldots s_n} \right) + \ii \Im\left( A_{s_1s_2\ldots s_n} \right)   -\mathrm{RBM}^R_{n,6nk^R+n+k^R, k^R+1}(\{s_i\}) - i\mathrm{RBM}^I_{n,6nk^I+n+k^I, k^I+1}(\{s_i\})  |\\
&\leq \max_{\lbrace s_i \rbrace} |\Re\left( A_{s_1s_2\ldots s_n} \right)  -\mathrm{RBM}^R_{n,6nk^R+n+k^R, k^R+1}(\{s_i\}) |+ \max_{\lbrace s_i \rbrace}  |\Im\left( A_{s_1s_2\ldots s_n} \right)  -\mathrm{RBM}^I_{n,6nk^I+n+k^I, k^I+1}(\{s_i\}) |\\
&< \frac{\epsilon}{2}+  \frac{\epsilon}{2} = \epsilon
\end{split}
\end{eqnarray}
Thus any complex tensor $A_{s_1s_2\ldots s_n}$ can be approximated by  $\mathrm{RBM}^R_{n,6nk^R+n+k^R, k^R+1}(\{s_i\}) + i\mathrm{RBM}^I_{n,6nk^I+n+k^I, k^I+1}(\{s_i\}) $ up to an arbitrary given precision $\epsilon$. 

To complete the proof, we need to show that $\mathrm{RBM}^R_{n,6nk^R+n+k^R, k^R+1}(\{s_i\}) + i\mathrm{RBM}^I_{n,6nk^I+n+k^I, k^I+1}(\{s_i\}) $ is a 2-layer RBM. Denote $k=\max\{k^R, k^I\}$, then $\mathrm{RBM}^R_{n,6nk^R+n+k^R, k^R+1}(\{s_i\}) $ and  $\mathrm{RBM}^I_{n,6nk^I+n+k^I, k^I+1}(\{s_i\}) $ can both be  written as $\mathrm{RBM}^{R/I}_{n,6nk+n+k, k+1}(\{s_i\})$. Hence 
\begin{equation}\label{Eq.101}
\begin{split}
& \mathrm{RBM}^R_{n,6nk+n+k, k+1}(\{s_i\}) + i\mathrm{RBM}^I_{n,6nk+n+k, k+1}(\{s_i\}) \\
& = \sum_{\{h_j\},  \{g_l\}} \mathrm{RBM}^R_{n+k+1, 6nk + n+ k}(s_1,s_2,\ldots,s_n,h_0,h_1,\ldots,h_k) + i \mathrm{RBM}^I_{n+k+1, 6nk + n+ k}(s_1,s_2,\ldots,s_n,h_0,h_1,\ldots,h_k)\\
&=\sum_{\substack{\{h_j\},  \{g_l\}}}  \exp (W^{sgR}_{il}s_i g_l+ W^{hgR}_{jl}h_j g_l + a^{sR}_{i}s_i+ a^{hR}_{j} h_j+ b^R_l g_l) + \exp (W^{sgI}_{il}s_i g_l+ W^{hgI}_{jl}h_j g_l + a^{sI}_{i}s_i+ a^{hI}_{j} h_j+ b^I_l g_l+ \frac{\pi i}{2})\\
&= \sum_{\substack{\{h_j\},  \{g_l\}}}  \exp (W^{sgR}_{il}s_i g_l+ W^{hgR}_{jl}h_j g_l + a^{sR}_{i}s_i+ a^{hR}_{j} h_j+ b^R_l g_l )\\&~~~\times  \bigg(1+ \exp (W^{sgI}_{il}s_i g_l+ W^{hgI}_{jl}h_j g_l + a^{sI}_{i}s_i+ a^{hI}_{j} h_j+ b^I_l g_l+ \frac{\pi i}{2}- W^{sgR}_{il}s_i g_l- W^{hgR}_{jl}h_j g_l - a^{sR}_{i}s_i- a^{hR}_{j} h_j- b^R_l g_l)\bigg)\\
&= \sum_{\substack{\{h_j\},  \{g_l\}}} \sum_{g_0=0}^{1}  \exp (W^{sgR}_{il}s_i g_l+ W^{hgR}_{jl}h_j g_l + a^{sR}_{i}s_i+ a^{hR}_{j} h_j+ b^R_l g_l )\\&~~~\times   \exp\bigg(g_0 (W^{sgI}_{il}s_i g_l+ W^{hgI}_{jl}h_j g_l + a^{sI}_{i}s_i+ a^{hI}_{j} h_j+ b^I_l g_l+ \frac{\pi i}{2}- W^{sgR}_{il}s_i g_l- W^{hgR}_{jl}h_j g_l - a^{sR}_{i}s_i- a^{hR}_{j} h_j- b^R_l g_l)\bigg)\\
&=\sum_{\substack{\{h_j\},  \{g_l\}, g_0}} \prod_{il} \mathrm{RBM}_{2,1}\left(s_i, g_l\right) \prod_{jl} \mathrm{RBM}_{2,1}\left( h_j, g_l\right)
\prod_{il} \mathrm{RBM}_{3,5}\left(g_0, s_i, g_l\right) \prod_{jl} \mathrm{RBM}_{3,5}\left(g_0,h_j, g_l\right) \\&~~~~~~~~~~~~~\times \prod_{i} \mathrm{RBM}_{2,1}\left( g_0, s_i\right)  \prod_{j} \mathrm{RBM}_{2,1}\left( g_0, h_j\right)	 \prod_{l} \mathrm{RBM}_{2,1}\left( g_0, g_l\right)\\
&=\sum_{\substack{\{h_j\},  \{g_l\}, g_0}}  \mathrm{RBM}_{6nk+2n+2k+2, 36kn^2+6n^2+36k^2 n+ 54k n+ 8n+6k^2+8k+1 }(\{s_i\},\{h_j\},\{g_l\}, g_0)
\end{split}
\end{equation}
In the above, we used $\{s_i\}$ to label $\{s_1, ..., s_n\}$, $\{h_j\}$ to label $\{h_0,h_1, ...,h_k \}$, $\{g_l\}$ to label $\{g_1, ..., g_{6nk+n+k}\}$. 
In the last second equality, we have used 
\begin{eqnarray}
\begin{split}
\exp(W^{sgR}_{il}s_i g_l + a_{i}^{sR} s_i ) &= \mathrm{RBM}_{2,1}\left(s_i, g_l\right)=\\
\exp( W^{hgR}_{jl}h_j g_l + a^{hR}_{j} h_j+ b^R_l g_l )&= \mathrm{RBM}_{2,1}\left( h_j, g_l\right)\\
\exp((W^{sgI}_{il}-W^{sgR}_{il})g_0 s_i g_l )&= \mathrm{RBM}_{3,5}\left(g_0, s_i, g_l\right)\\
\exp( (W^{hgI}_{jl}-W^{hgR}_{jl})g_0, h_j g_l)&= \mathrm{RBM}_{3,5}\left(g_0,h_j, g_l\right)\\
\exp( (a^{sI}_{i}-  a^{sR}_{i})g_0 s_i)&= \mathrm{RBM}_{2,1}\left( g_0, s_i\right)\\
\exp( (a^{hI}_{j}- a^{hR}_{j}) g_0h_j)&= \mathrm{RBM}_{2,1}\left( g_0, h_j\right)	\\
\exp((b^I_l-b^R_l) g_0 g_l+ \frac{\pi i}{2} g_0)&=\mathrm{RBM}_{2,1}\left( g_0, g_l\right)\\
\end{split}
\end{eqnarray}
From the last second line of Eq.~\eqref{Eq.101} and using the property \ref{property.RBM_notation}, we can simplify the product of 1-layer RBM's to a single RBM, which is a last line. Explicitly, the visible spins in the 1-layer RBM is the union $\{s_i\}\cup \{h_j\}\cup \{g_l\}\cup \{g_0\}$ whose number is $6nk+2n+2k+2$. The number of hidden spins in the last line comes from the following sum:
\begin{equation}\label{Eq.100}
[n(6nk+n+k)]+[ (k+1)(6nk+n+k)]+ [5n(6nk+n+k)]+[5 (k+1)(6nk+n+k)] + [n]+ [k+1] + [6nk+n+k]
\end{equation}
The term in each of the 7 square brackets Eq.~\eqref{Eq.100} counts the number of hidden spins contributed from each of the 7 products in the last second equality of Eq.~\eqref{Eq.101}.

Finally, we reorganize the sum of 1-layer RBM's into a 2-layer RBM, where there are $n$ visible spins $\{s_i\}$, $36kn^2+6n^2+36k^2 n+ 54k n+ 8n+6k^2+8k+1$ hidden spins in the first hidden layer, and $6nk+n+2k+2$ hidden spins in the second layer. We denote this 2-layer RBM as 
\begin{eqnarray}\label{Eq.104}
\begin{split}
&\mathrm{RBM}_{n, 36kn^2+6n^2+36k^2 n+ 54k n+ 8n+6k^2+8k+1, 6nk+n+2k+2}(\{s_i\})\\
&\equiv  \sum_{\substack{\{h_j\},  \{g_l\}, g_0}}  \mathrm{RBM}_{6nk+2n+2k+2, 36kn^2+6n^2+36k^2 n+ 54k n+ 8n+6k^2+8k+1 }(\{s_i\},\{h_j\},\{g_l\}, g_0)
\end{split}
\end{eqnarray}

In summary, for any given $\epsilon>0$ and any complex valued tensor $A_{s_1s_2\ldots s_n} $, we find a 2-layer RBM to approximate it within the precision $\epsilon$, 
\begin{eqnarray}\label{Eq.103}
\max_{\lbrace s_i \rbrace} | A_{s_1s_2\ldots s_n} -\mathrm{RBM}_{n, 36kn^2+6n^2+36k^2 n+ 54k n+ 8n+6k^2+8k+1, 6nk+n+2k+2}(\{s_i\}) | \le \epsilon.
\end{eqnarray}
The 2-layer RBM is graphically represented as follows. 
\begin{equation}\label{Eq.3}
\begin{gathered}
\includegraphics[width=10cm]{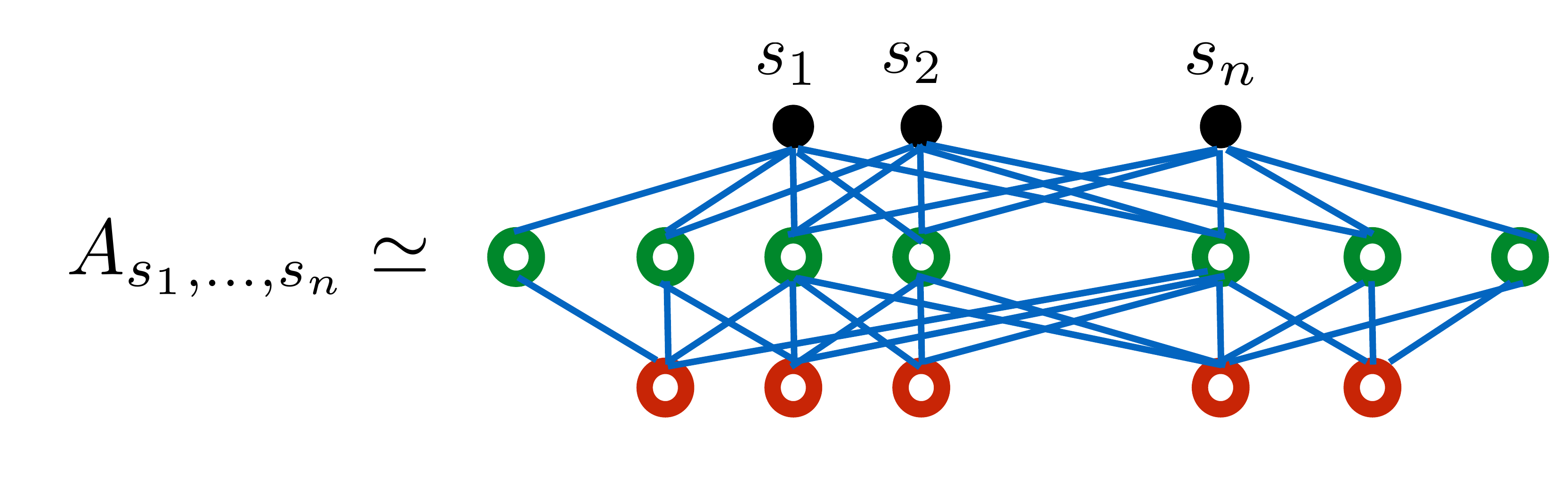}
\end{gathered}
\end{equation}
where $\simeq$ represents the approximation up to precision $\epsilon$, as in  Eq.~\eqref{Eq.103}.  The black dots represent the visible spins $\{s_1, ..., s_n\}$. There are  $6nk+n+2k+2$ red circles in the second hidden layer which  represent the set of $6nk+n+2k+2$ spins $\{h_j\}\cup \{g_l\}\cup \{g_0\}$ in Eq.~\eqref{Eq.104}. There are $36kn^2+6n^2+36k^2 n+ 54k n+ 8n+6k^2+8k+1$  green circles in the first hidden layer which represent the $36kn^2+6n^2+36k^2 n+ 54k n+ 8n+6k^2+8k+1$ hidden spins in the first hidden layer. This completes the proof. 
\end{proof}

Theorem \ref{Theorem1} gives an approximation of a complex tensor with $n$ indices, $A_{s_1, ..., s_n}$, by a 2-layer RBM. We further demand the complex tensor to be a single MPS tensor $A^S_{LR}$ and apply theorem \ref{Theorem1} to $A^S_{LR}$. 
An MPS tensor $A^S_{LR}$ has three sets of indices: the set of physical indices $S=\{s_1, ..., s_n\}$, the set of left virtual indices $L= \{l_1, ..., l_m\}$ and the set of right virtual indices $R=\{r_1, ..., r_m\}$, where we assume the number of left and right virtual indices are equal. To represent $A^S_{LR}$ as an RBM, we regard both $S, L$ and $R$ as the visible spins, and introduce additional hidden spins.  Applying Theorem \ref{Theorem1}, we find the following corollary:
\begin{corollary}\label{corollary.MPS_Tensor_2LayerRBM}
Let $A^S_{LR}=A^{s_1s_2\ldots s_n}_{l_1l_2\ldots l_m, r_1r_2 \ldots r_m} $ be a MPS tensor  with $n$ physical spins $\{s_i\}, i=1, ..., n$ and $m$ virtual spins $\{l_p\},\{r_p\}, p=1, ..., m $. $k^{R\pm}$ is the number of elements of the real part of $A^S_{LR}$ which are strictly positive and negative respectively, and $k^{I\pm}$ is the number of elements of the imaginary part of $A^S_{LR}$ which are strictly positive and negative respectively. Denote $k=\max(k^{R+},k^{R-},k^{I+},k^{I-})$. Then given any positive number $\epsilon$, $A^S_{LR}$ can be approximated by a 2-layer RBM within the uncertainty $\epsilon$
\begin{equation}
\begin{split}
\max_{\{s_i\}, \{l_p\}, \{r_p\}} |A^{s_1s_2\ldots s_n}_{l_1l_2\ldots l_m, r_1r_2 \ldots r_m}- \mathrm{RBM}_{n+2m,H_1 , H_2}(\{s_i\},\{l_p\}, \{r_p\})| <\epsilon
\end{split}
\end{equation}
where $H_1=36k(n+2m)^2+6(n+2m)^2+36k^2 (n+2m)+ 54k (n+2m)+ 8(n+2m)+6k^2+8k+1$ is the number of hidden spins in the first hidden layer,  and $H_2=6(n+2m)k+(n+2m)+2k+2$  is the number of hidden spins in the second hidden layer. 
\end{corollary}

\begin{proof}
To prove the corollary \ref{corollary.MPS_Tensor_2LayerRBM}, we directly apply Theorem \ref{Theorem1}. We replace the indices $\{s_i\}$ in Eq.~\eqref{Eq.103} with the indices $\{s_i\}\cup \{l_p\}\cup \{r_p\}$. Then $A_{s_1...s_n}$ is replaced by $A^{s_1s_2\ldots s_n}_{l_1l_2\ldots l_m, r_1r_2 \ldots r_m}$, and correspondingly, $n$ is replaced by $n+2m$ because the number of indices of $A$ tensor is now $n+2m$. 

It is helpful to represent the 2-layer RBM graphically. We first simplify the graphical representation in Eq.~\eqref{Eq.3} as 
\begin{equation}\label{Eq.1000}
\begin{gathered}
\includegraphics[width=7cm]{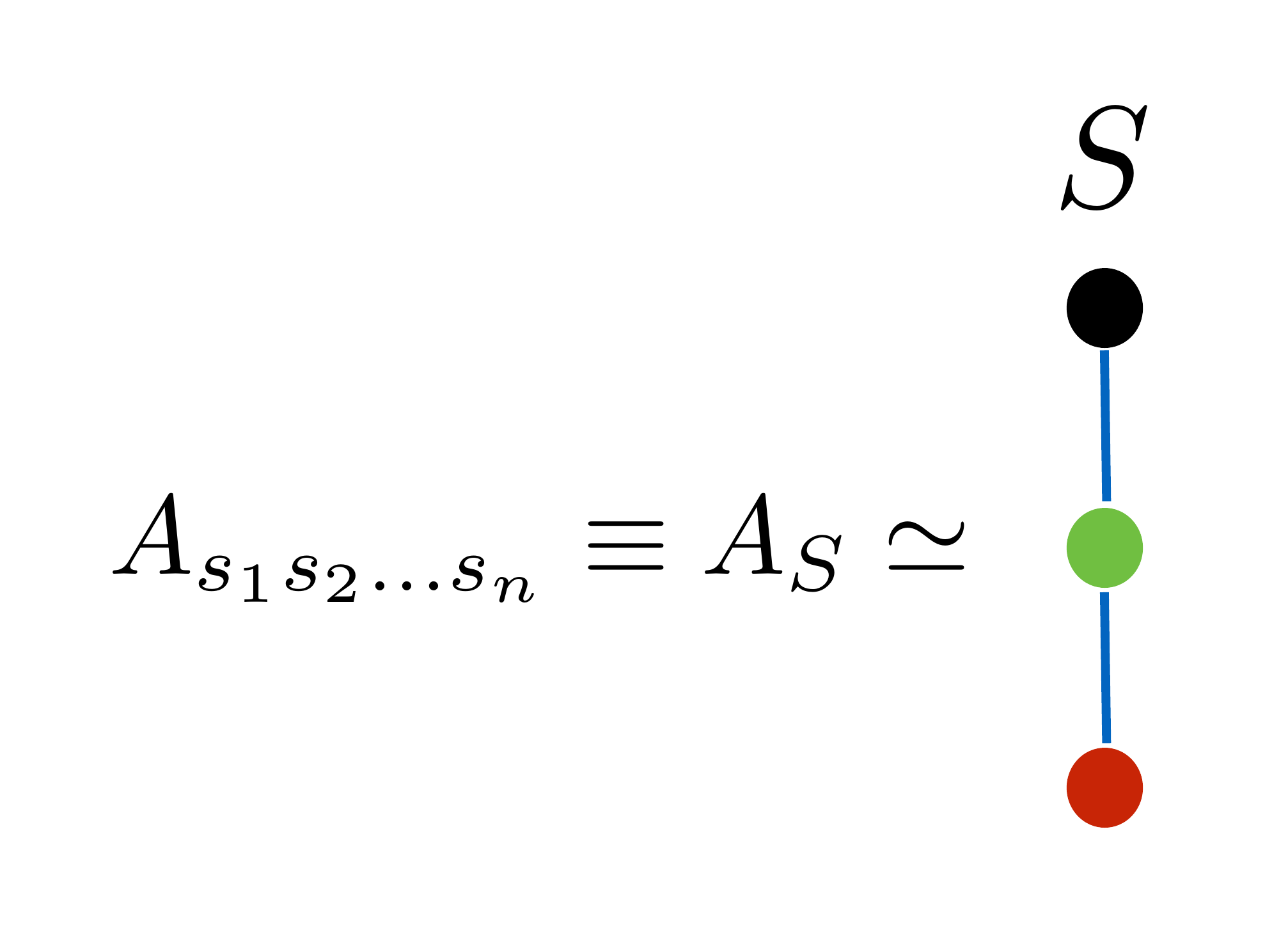}
\end{gathered}
\end{equation}
where we schematically use a black dot to represent the collective visible spins $S$, and use a single green dot and red dot to represent the collective hidden spins in the first and second hidden layers respectively. Using this simplified notation, we can represent the tensor with physical spins $S$ and virtual spins $L, R$ as
\begin{equation}\label{Eq.1001}
\includegraphics[width=10cm]{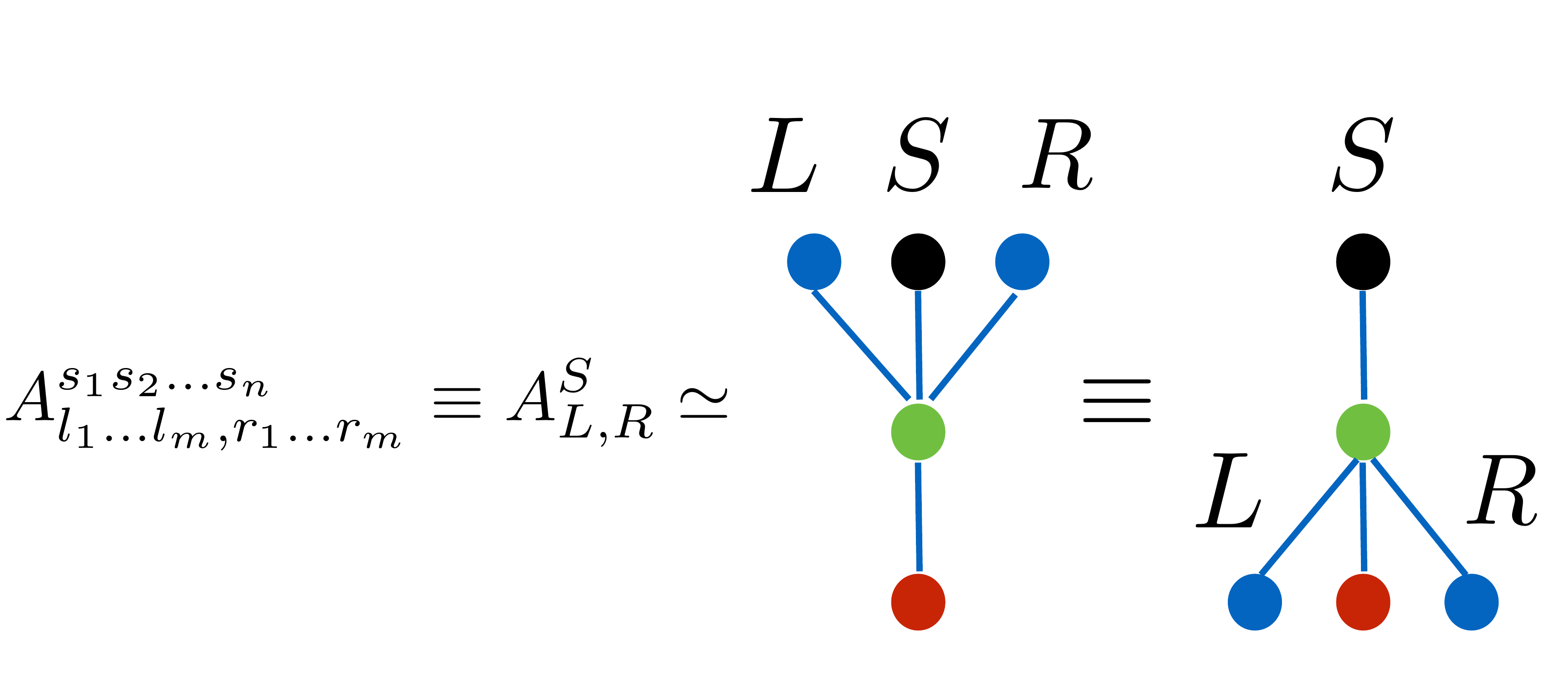}
\end{equation}
where the blue dots represent the virtual spins $L$ and $R$. In the last equality of Eq.~\eqref{Eq.1001}, we fold the spins $L$ and $R$ from the visible layer to the second hidden layer because $L$ and $R$ are to be summed over when the $A^S_{L,R}$ tensors are contracted. 
\end{proof}

A MPS wave function can be obtained by contracting the virtual indices of the MPS tensors $A_{L,R}^S$. Contracting the virtual indices of the MPS tensor amounts to contracting the hidden spins in the second hidden layer of the RBM. Applying Corollary \ref{corollary.MPS_Tensor_2LayerRBM}, we find that any MPS wave function can be represented by a 2-layer RBM neural network. To streamline the analysis, we first present a lemma
\begin{lemma}\label{lemma.final}
	Suppose $T^{S_x}_{L_x L_{x+1}}$ and $R^{S_x}_{L_x L_{x+1}}$ are two tensors, which satisfy
	\begin{eqnarray}\label{single}
	\max_{L_x, L_{x+1}, S_x} |T^{S_x}_{L_x L_{x+1}}- R^{S_x}_{L_x L_{x+1}}|< \epsilon
	\end{eqnarray}
	for $\epsilon>0$. Denote the {normalized} wave functions 
	\begin{eqnarray}
	\Psi_T(\{S_x\})= \frac{\Tr(\prod_{x=1}^N T^{S_x}) }{\sum_{\{S_x\}} |\Tr(\prod_{x=1}^N T^{S_x}) |^2}, ~~~~~ 	\Psi_R(\{S_x\})= \frac{\Tr(\prod_{x=1}^N R^{S_x}) }{\sum_{\{S_x\}} |\Tr(\prod_{x=1}^N R^{S_x}) |^2}
	\end{eqnarray}
	Then the norm between the two wave functions, to the linear order of $\epsilon$, is given by
	\begin{eqnarray}\label{orderN}
	|\Psi_T(\{S_x\}) - \Psi_R(\{S_x\})|\simeq \mathcal{O}(N)\epsilon
	\end{eqnarray}
	where $\mathcal{O}(N)$ is a number of order $N$.

\end{lemma}

\begin{proof}
	Using Eq.~\eqref{single}, one can write $R$ as 
	\begin{eqnarray}\label{expan}
	R^{S_x}_{L_x L_{x+1}} = T^{S_x}_{L_x L_{x+1}} + \epsilon U^{S_x}_{L_x L_{x+1}} + \mathcal{O}(\epsilon^2)
	\end{eqnarray}
	where $U^{S_x}_{L_x L_{x+1}}$ is an arbitrary function whose maximal norm element is of order 1. Substituting Eq.~\eqref{expan} into $|\Psi_T(\{S_x\}) - \Psi_R(\{S_x\})|$, one gets 
	\begin{eqnarray}\label{order1}
	\begin{split}
	&|\Psi_T(\{S_x\}) - \Psi_R(\{S_x\})|= \left|\frac{\Tr(\prod_{x=1}^N T^{S_x}) }{\sum_{\{S_x\}} |\Tr(\prod_{x=1}^N T^{S_x}) |^2}- \frac{\Tr(\prod_{x=1}^N R^{S_x}) }{\sum_{\{S_x\}} |\Tr(\prod_{x=1}^N R^{S_x}) |^2}\right|\\
	&=  \left|\frac{\Tr(\prod_{x=1}^N T^{S_x}) }{\sum_{\{S_x\}} |\Tr(\prod_{x=1}^N T^{S_x}) |^2}- \frac{\Tr(\prod_{x=1}^N(T^{S_x} + \epsilon U^{S_x} + \mathcal{O}(\epsilon^2))) }{\sum_{\{S_x\}} |\Tr(\prod_{x=1}^N (T^{S_x} + \epsilon U^{S_x} + \mathcal{O}(\epsilon^2))) |^2}\right|\\
	&\simeq  \left|\frac{\Tr(\prod_{x=1}^N T^{S_x}) }{\sum_{\{S_x\}} |\Tr(\prod_{x=1}^N T^{S_x}) |^2}- \frac{\Tr(\prod_{x=1}^N T^{S_x} ) + \epsilon\sum_{x=1}^N \Tr( T^{S_1}T^{S_2}... U^{S_x}...T^{S_N} )}{ \sum_{\{S_x\}} |\Tr(\prod_{x=1}^N T^{S_x}) |^2 \bigg(1+ \frac{\epsilon\sum_{x=1}^N \Tr( T^{S_1}T^{S_2}... U^{S_x}...T^{S_N} ) |\Tr(\prod_{x=1}^N T^{S_x}) |^2+ c.c}{ \sum_{\{S_x\}} |\Tr(\prod_{x=1}^N T^{S_x}) |^2}\bigg)}\right|\\
	&=\left| -\epsilon \frac{\sum_{x=1}^N \Tr( T^{S_1}T^{S_2}... U^{S_x}...T^{S_N} )}{\sum_{\{S_x\}} |\Tr(\prod_{x=1}^N T^{S_x}) |^2} +\epsilon \frac{\Tr(\prod_{x=1}^N T^{S_x} ) (\sum_{x=1}^N \Tr( T^{S_1}T^{S_2}... U^{S_x}...T^{S_N} ) |\Tr(\prod_{x=1}^N T^{S_x}) |^2+ c.c)}{(\sum_{\{S_x\}} |\Tr(\prod_{x=1}^N T^{S_x}) |^2 )^2} \right|\\
	&= \mathcal{O}(N) \epsilon
	\end{split}
	\end{eqnarray}
	In the third line, we only keep the terms linear in $\epsilon$. In the last equation, we used the fact that
	\begin{equation}
	\frac{ \Tr( T^{S_1}T^{S_2}... U^{S_x}...T^{S_N} )}{\sum_{\{S_x\}} |\Tr(\prod_{x=1}^N T^{S_x}) |^2}\sim \mathcal{O}(1), ~~~~~ \frac{\Tr(\prod_{x=1}^N T^{S_x} ) (\Tr( T^{S_1}T^{S_2}... U^{S_x}...T^{S_N} ) |\Tr(\prod_{x=1}^N T^{S_x}) |^2+ c.c)}{(\sum_{\{S_x\}} |\Tr(\prod_{x=1}^N T^{S_x}) |^2 )^2}\sim \mathcal{O}(1)
	\end{equation}
	due to normalization. In the last equality of Eq.~\eqref{order1}, there are $N$ such order 1 coefficients of $\epsilon$, hence the total coefficient of $\epsilon$ is $\mathcal{O}(N)$. This completes the proof. 
\end{proof}

\begin{theorem}\label{theorem.Final}
Suppose $\Psi(\{S_x\})$ is an arbitrary normalized MPS wave function
\begin{eqnarray}
\Psi(\{S_x\})=\frac{\Tr(\prod_{x=1}^N A^{S_x})}{\sum_{\{S_x\}} |\Tr(\prod_{x=1}^N A^{S_x})|^2}
\end{eqnarray} 
where $A^{S_x}_{L_x, L_{x+1}}$ is the MPS tensor and $N$ is the number of unit cells. In each unit cell, there are $n$ physical spins $S_x\in \{0,1\}^n$ and $m$ virtual spins $L_x\in \{0,1\}^m$.  For any positive number $\epsilon$, there exists a 2-layer RBM, $\mathrm{RBM}_{nN,  H_1N,  H_2N+ m N}(\{S_x\})$ such that their error is linear in $N$, 
\begin{equation}\label{Eq.1006}
    \max_{\{S_x\}} |\Psi(\{S_x\})-\Psi_{\mathrm{RBM}}(\{S_x\}) |\simeq \mathcal{O}(N) \epsilon
\end{equation}
where
\begin{eqnarray}\label{eqd66}
\Psi_{\mathrm{RBM}}(\{S_x\})= \frac{\mathrm{RBM}_{nN,  H_1N,  H_2N+ m N}(\{S_x\})}{\sum_{\{S_x\}} |\mathrm{RBM}_{nN,  H_1N,  H_2N+ m N}(\{S_x\})|^2}
\end{eqnarray}
is the normalized RBM. 
Here $H_1$ is the number of hidden spins in the first hidden layer per unit cell, and $H_2$ is the number of hidden spins (not including $\{L_x\}$) in the second hidden layer per unit cell. The precise values of $H_1, H_2$ are  given in Corollary \ref{corollary.MPS_Tensor_2LayerRBM}. $\mathcal{A}=\max_{ S_{x}, L_{x}, L_{x+1}} |A^{S_x}_{L_{x}, L_{x+1}}|$. A graphical representation of the RBM representation is given as follows
\begin{equation}
\includegraphics[width=0.5\textwidth]{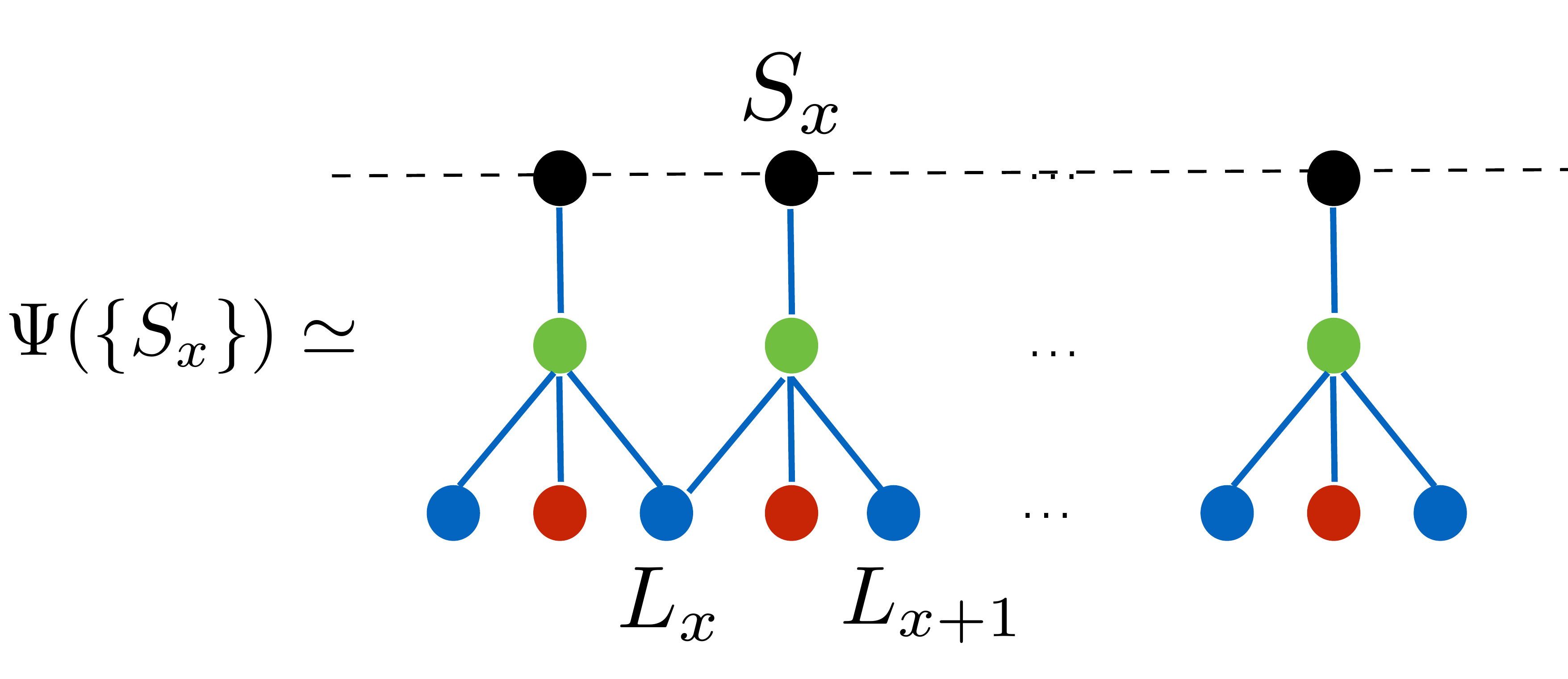}
\end{equation}
where $\simeq$ is short for the approximation in Eq.~\eqref{Eq.1006}.
\end{theorem}

\begin{proof}
This theorem follows from the Corollary \ref{corollary.MPS_Tensor_2LayerRBM} by contracting the virtual spins $\{L_x\}$. From Corollary \ref{corollary.MPS_Tensor_2LayerRBM}, we have 
\begin{equation}\label{Eq.1002}
    \max_{S_x, L_x, L_{x+1}}|A^{S_x}_{L_x, L_{x+1}}- \mathrm{RBM}_{n+2m, H_1, H_2}(S_x, L_x, L_{x+1})|<\epsilon, ~~~~\forall x=1, ..., N
\end{equation}
Then we apply lemma \ref{lemma.final}, where we demand the pair of tensors $(T^{S_x}_{L_x, L_{x+1}}, R^{S_x}_{L_x, L_{x+1}})$ in lemma \ref{lemma.final} to be the pair of tensors $(A^{S_x}_{L_x, L_{x+1}}, A^{S_x}_{L_x, L_{x+1}}, \mathrm{RBM}_{n+2m, H_1, H_2}(S_x, L_x, L_{x+1}))$. Substituting into the Eq.~\eqref{orderN} in lemma \ref{lemma.final}, Eq.~\eqref{eqd66} immediately follows. This completes the proof.

\end{proof}

\end{document}